\documentclass[11pt]{amsart}
\usepackage{hyperref}

\usepackage{amsfonts,amsmath,amssymb}
\usepackage[margin=1in]{geometry}

\usepackage{algorithm}
\usepackage[noend]{algpseudocode}

\usepackage{graphicx}
\newtheorem{theorem}{Theorem}[section]
\newtheorem{proposition}[theorem]{Proposition}
\newtheorem{lemma}[theorem]{Lemma}
\newtheorem{corollary}[theorem]{Corollary}
\newtheorem{definition}[theorem]{Definition}

\newcommand{\bE}{\ensuremath{\mathbf{E}}}

\DeclareMathOperator{\poly}{poly}
\DeclareMathOperator{\polylog}{polylog}

\begin{document}

\author[David G. Harris]{David G. Harris$^1$}
\setcounter{footnote}{0}
\addtocounter{footnote}{1}
\footnotetext{Department of Computer Science, University of Maryland, 
College Park, MD 20742. 
Research supported in part by NSF Awards CNS-1010789 and CCF-1422569.
Email: \texttt{davidgharris29@gmail.com}.}

\title[Fooling polylogarithmic juntas and the Lov\'{a}sz Local Lemma]{Deterministic parallel algorithms for fooling polylogarithmic juntas and the Lov\'{a}sz Local Lemma}

\begin{abstract}
Many randomized algorithms can be derandomized efficiently using either the method of conditional expectations or probability spaces with low (almost-) independence. A series of papers, beginning with Luby (1993) and continuing with Berger \& Rompel (1991) and Chari et al. (2000), showed that these techniques can be combined to give deterministic parallel algorithms for combinatorial optimization problems involving sums of $w$-juntas. We improve these algorithms through derandomized variable partitioning, reducing the processor complexity to essentially independent of $w$ and time complexity to linear in $w$.

As a key subroutine, we give a new algorithm to generate a probability space which can fool a given set of neighborhoods. Schulman (1992) gave an NC algorithm to do so for neighborhoods of size $w \leq O(\log n)$. Our new algorithm is in $\text{NC}^1$, with essentially optimal time and processor complexity, when $w = O(\log n)$; it remains in NC  up to $w = \polylog(n)$. This answers an open problem of Schulman.

One major application of these algorithms is an NC algorithm for the Lov\'{a}sz Local Lemma. Previous NC algorithms, including the seminal algorithm of Moser \& Tardos (2010) and the work of Chandrasekaran et. al (2013), required that (essentially) the bad-events could span only $O(\log n)$ variables; we relax this to $\polylog(n)$ variables. We use this for an $\text{NC}^2$ algorithm for defective vertex coloring, which works for arbitrary degree graphs.
\end{abstract}

\maketitle

This is an extended version of a paper appearing in the Proceedings of the 28th ACM-SIAM Symposium on Discrete Algorithms (SODA) 2017.

\section{Introduction}
Many algorithms can be formulated as optimization problems, in which we seek to maximize or minimize a function of the form $S(x) = \sum_j f_j(x)$ over $x \in \{0, 1 \}^n$; we refer to the summands $f_j$ as \emph{objective functions}. These may correspond to a scoring function measuring solution quality, or they might be indicators for bad events to avoid.  We will consider cases in which each $f_j$ depends on at most $w$ coordinates of $x$; this is known as a \emph{$w$-junta}.  

This often leads to randomized algorithms with the following structure: if $X$ is drawn from a suitable distribution (say independent fair coins), then $\bE[S(X)] = \sum_j\bE[f_j(X)] = S_0$. Obviously, there exists some $x \in \{0, 1 \}^n$ with the property $S(x) \geq S_0$. Usually we can find such an $x$ with a randomized algorithm, since a ``typical'' vector $x$ has the property $S(x) \approx S_0$. 

A key derandomization problem is thus to find such $x$ deterministically. There are two main paradigms to do so: conditional expectations and low-independence probability spaces. To use conditional expectations, we gradually assign the bits of $X$ to $0$ or $1$, ensuring that at each step the conditional expectation $\bE[S(X)]$ increases. To solve this by low-independence, we draw the random variables $X$ from a probability space which has the same $w$-wise marginal distributions as the independent space $\{0, 1 \}^n$. Each of these methods has disadvantages. The method of conditional expectations is inherently sequential: decisions about some $x_i$ depend on the assignment of $x_1, \dots, x_{i-1}$.  The method of low-independence can easily be parallelized, but leads to large processor counts as each element of the probability space requires a separate processor.

A hybrid approach was proposed by Luby in \cite{luby-old} and extended by Berger \& Rompel \cite{berger-simulating}, which combines parallelism with low processor complexity. Their key observation is that probability spaces with polylog-wise-independence can be described as linear codes over $GF(2)$, of length $\polylog(n)$. The method of conditional expectations can be applied to the code itself, not directly to the solution vector.  The main limitation of this hybrid algorithm is that, at least in its simplest form, it has processor complexity which is \emph{exponential} in $w$. Berger \& Rompel describe a limited number of problem-specific techniques to overcome this.  In this paper, we will investigate a more general method of dealing with this computational bottleneck, based on a derandomization of random variable partitioning.

\subsection{Alternate derandomization approaches} We mention three other general approaches to derandomization, and the ways in which they fail to cover some key applications. The first approach is to use a probability space which is $\epsilon$-approximately $w$-wise-independent (see Definition~\ref{approx-def}). Such a space is significantly smaller than a fully-independent space. If the objective functions $f_j$ were simply monomial functions, or more generally had small decision tree complexity, then their expectation would differ only slightly between an $\epsilon$-approximately independent and a fully independent space. However, in many applications, $f_j$ may be significantly more complex and the overall bias can become as large as $2^w \epsilon$ --- requiring $\epsilon$ to be super-polynomially small for $w = \omega(\log n)$, and requiring the probability space to be too large to explore exhaustively.

A related approach is one of Schulman \cite{schulman} for generating a probability space which fools a given list of neighborhoods (see Section~\ref{fourier-sec} for a formal definition). If $\Omega$ fools the neighborhoods corresponding to each $f_j$, then there is guaranteed to exist some $x \in \Omega$ with $S(x) \geq S_0$; if $\Omega$ has small support then this leads to an efficient algorithm. Although the space $\Omega$ can be significantly smaller than a fully $w$-wise-independent space, it is still super-polynomial for $w = \omega(\log n)$ and so this approach does not give NC algorithms. Fooling neighborhoods will nonetheless be a key building block of our algorithms.
 
Finally, the derandomization technique of Sivakumar \cite{sivakumar} can be applied when the functions $f_j$ are computed via automata on a polynomial-sized state-space. One can build a relatively small probability distribution which fools a polynomial number of such automata. However, one critical aspect of Sivakumar's method is that the multiple automata must all process the input bits \emph{in the same order.} Many applications lack this property, most notably the algorithms for the Lov\'{a}sz Local Lemma. Another disadvantage of Sivakumar's method is its high processor complexity (on the order of $n^{20}$ or more). 

\subsection{Our contributions and overview}
In Section~\ref{r1fAheavy-codeword}, we present a new algorithm to produce probability spaces fooling a list of neighborhood or a list of Fourier characters over $GF(2)$ (a closely related problem). These are important subroutines needed for the algorithmic approach of Berger \& Rompel \cite{berger-simulating}. The algorithm we develop has significantly lower complexity than previous algorithms. In particular, when the neighborhood size $w$ is  $w = O(\log n)$, then we obtain an $\text{NC}^1$ algorithm and when $w = \polylog(n)$ we obtain an NC algorithm.

In addition to their use in the Berger-Rompel framework, these algorithms can be used for some other derandomization problems. For instance, we obtain a near-optimal algorithm to find a codeword of Hamming weight at least $L/2$ in a length-$L$ binary code, a toy derandomization problem introduced by \cite{naor2}.

In Section~\ref{r1fmain-alg-sec}, we consider fooling sums of $w$-juntas. As we have discussed, the main bottleneck in the Berger-Rompel algorithm \cite{berger-simulating} is the exponential processor dependency on $w$. We give an algorithm based on random variable partitioning, which is then derandomized. This approach makes the processor complexity independent of $w$ while giving a \emph{linear} time dependency on $w$. This allows us to handle, for the first time in NC, many applications with $w = \polylog (n)$.

We describe a sample application in Section~\ref{r1frainbow-sec} to rainbow hypergraph coloring. Given a $d$-uniform hypergraph on $m$ edges, we wish to $d$-color the vertices so that at least $m d!/d^d$ edges see all $d$ colors (as is expected in a uniform random coloring). This was an example application given in \cite{berger-simulating}. Although \cite{berger-simulating}  did not provide concrete complexity bounds, their algorithm appears to require $O(\log^4 m)$ time and $O(m^{1 + \log 2})$ processors. We reduce this to roughly $\tilde O(\log^2 m )$ time and $O(m)$ processors. This illustrates how our derandomization procedure has been optimized for processor and time complexity, so that it can be beneficial even for applications with prior NC algorithms.

In Section~\ref{r1fmt-sec} we consider the seminal Moser-Tardos algorithm for the Lov\'{a}sz Local Lemma \cite{mt}. In this setting, one seeks to avoid a set of ``bad events'', which are boolean functions of the variables $x_1, \dots, x_n$. There have been some NC versions of this algorithm, appearing in the original paper of Moser \& Tardos along with some extensions in \cite{mt2, mt3}. These algorithms are somewhat limited in the types of problems they can handle, with restrictive conditions on the decision-tree-complexity of the bad-events. We greatly expand the scope of these algorithms to give NC algorithms in almost any application where the bad-events are $w$-juntas for $w = \polylog(n)$. 

In Section~\ref{r1fmt-example}, we apply our LLL derandomization to two graph theory applications. First,  \emph{defective vertex coloring}: given a graph $G$ of maximum degree $\Delta$, we achieve an $\tilde O(\log^2 n)$-time algorithm for a $k$-defective vertex coloring with $c = O(\Delta/k)$ colors. Notably, although our general LLL algorithm only applies to bad-events which span a polylogarithmic number of variables (in particular here, when $\Delta \leq \polylog(n)$), our coloring algorithm works for arbitrary values of $\Delta$ and $k$. The second application is to domatic partition; here we only get an NC algorithm for graphs of degree $k = \text{polylog}(n)$.

\subsection{Notation and conventions}
All our algorithms will be described in the deterministic EREW PRAM model. In this model, we say an algorithm $\mathcal A$ has \emph{complexity} $(C_1, C_2)$ if it runs in $C_1$ time and $C_2$ processors. In order to focus on the leading-order terms, we often use a looser metric which we refer to as \emph{quasi-complexity}. We say $\mathcal A$ has \emph{quasi-complexity} $(C_1, C_2)$ if it has complexity $( \tilde O(C_1), C_2^{1+o(1)})$, where we define $\tilde O(t) = t (\log t)^{O(1)}$.\footnote{It is very difficult to obtain estimates which are finer than this; small changes in the computational model or the input data (for example, the register size, the precisions of the real-valued weights, or the atomic arithmetic operations) can change the runtime by hard-to-track polyloglog factors.}

We let $[n]$ denote the set $\{1, \dots, n \}$. For any collection of sets $E \subseteq 2^{[n]}$, we define the \emph{width} of $E$ as $\text{width}(E) = \max_{e \in E} |e|$. For a probability space $\Omega$, we use the notation $|\Omega|$ to mean the cardinality of the support of $\Omega$. Given a set $X$, we write $x \sim X$ to mean that $x$ is drawn from the uniform probability distribution on $X$.  For a boolean predicate $P$, we use the Iverson notation where $[P]$ is one if $P$ is true and zero otherwise. 

We write $GF(q)$ for the finite field with $q$ elements. In particular, the field $GF(2^s)$ can be represented as $s$-bit binary vectors, and addition in the field is taken mod 2, the same as coordinatewise XOR. We write this addition operation as $\oplus$.

Throughout, $\log x$ refers to the natural logarithm and $\log_2 x$ to the base-two logarithm.

\section{Fooling neighborhoods}
\label{r1fAheavy-codeword}

\subsection{Fourier characters, neighborhoods, and codes}
\label{fourier-sec}
 Many probability spaces satisfying (approximate) independence conditions are built on top of codes over $GF(2)$. These are closely related to Fourier characters over $GF(2)$. We begin by reviewing some definitions and basic results.

\begin{definition}
A \emph{Fourier character over $GF(2)$} is a function $\chi_e: GF(2)^n \rightarrow \{-1, 1\}$ defined by $\chi_e(x) = (-1)^{\sum_{i \in e} x_i}$, for some $e \subseteq [n]$
\end{definition}

We say that a probability space $\Omega$ is \emph{unbiased for $e$} if $\bE_{X \sim \Omega}[\chi_e(X)] = \bE_{X \sim \{0,1 \}^n}[ \chi_e(X)]$. This condition trivially holds for $e = \emptyset$ (in which case $\chi_e(X) = 1$ with probability one). For $e \neq \emptyset$, we have $\bE_{X \sim \{0,1 \}^n}[ \chi_e(X) ] = 0$ and so the condition is that $\bE_{X \sim \Omega} [ \chi_e(X)] = 0$ as well. 

We say $\Omega$ \emph{fools $e$} if $\Omega$ is unbiased for every subset $f \subseteq e$. In this context, the set $e$ is referred to as a \emph{neighborhood} and this condition is also referred to as \emph{fooling neighborhood $e$.} Likewise, we say $\Omega$ is unbiased for (respectively, fools) a list $E = \{e_1, \dots, e_m \}$ if $\Omega$ is unbiased for (respectively, fools) each $e_1, \dots, e_m$. 

Our notation and definitions will differ slightly from the standard use in coding theory. Given a list of vectors $A(1), \dots, A(n) \in GF(2)^L$,  we refer to the list $A = A(1), \dots, A(n)$ as a \emph{code of length $L$ and size $n$}. We use the following notational shortcut throughout: if $x \in \{0, 1 \}^n$ and $e \subseteq [n]$ is a set, then we define $x(e) = \bigoplus_{i \in e} x_i$. If $A$ is a collection of $n$ binary vectors and $e \subseteq [n]$, then $A(e)$ is the binary vector defined coordinatewise by $\bigoplus_{i \in e} A(i)$.

\begin{proposition}
\label{r1ffourier-prop}
For any boolean function $g: \{0, 1 \}^n \rightarrow \mathbf R$, there are weights $\gamma_e \in \mathbf R$, where $e$ ranges over all $2^n$ subsets of $[n]$, such that for all $x \in \{0, 1 \}^n$ we have
$$
g(x) = \sum_{e \subseteq [n]} \gamma_e \chi_e(x)
$$

The weights $\gamma$ can be determined with quasi-complexity $(n, 2^n)$.
\end{proposition}
\begin{proof}
This is the Discrete Fourier Transform over $GF(2)$. Define $\gamma_e = 2^{-n} \sum_{y \in \{0,1 \}^n} \chi_e(y) g(y)$; these weights can be computed efficiently using the well-known Fast Walsh-Hadamard Transform algorithm.
\end{proof}

\begin{proposition}
\label{fool-nbhd-prop}
If a probability space $\Omega$ fools the neighborhood $e \subseteq [n]$, then for all $z \in \{0, 1 \}^n$ we have
$$
P_{x\sim \Omega}( \bigwedge_{i \in e} x_i = z_i) = 2^{-|e|}
$$
\end{proposition}
\begin{proof}
Let $g(x) = [ \bigwedge_{i \in e} x_i = z_i ]$. By Proposition~\ref{r1ffourier-prop}, there exist weights $\gamma_f$, where $f$ ranges over subsets of $e$, such that $g(x) = \sum_{f \subseteq e} \gamma_f \chi_f(x)$. Then
$$
\bE_{X \sim \Omega}[g(X)] = \sum_{f \subseteq e} \gamma_f \bE_{X \sim \Omega} [ \chi_f(X) ] = \sum_{f \subseteq e} \gamma_f \bE_{X \sim \{0,1 \}^n} [\chi_f(X)] = \bE_{X \sim \{0,1 \}^n}[ g(X) ] = 2^{-|e|}.
$$
\end{proof}

The main connection between codes, Fourier characters, and fooling neighborhoods comes from the following construction:
\begin{definition}
Given a code $A$ of length $L$, define the probability space $\Omega^A$ as follows: draw a vector $y \sim GF(2)^L$, and set $X_i = A(i) \cdot y$ for $i=1, \dots, n$, where $\cdot$ is the inner product over $GF(2)^L$. Note that $|\Omega^A| = 2^L$.
\end{definition}

\begin{definition}[$E$-unbiased code]
The code $A$ is an \emph{$E$-unbiased code} if $A(e) \neq \vec 0$ for all non-empty sets $e \in E$. 
\end{definition}

\begin{proposition}
If $A$ is an $E$-unbiased code, then $\Omega^A$ is unbiased for every $e \in E$.
\end{proposition}
\begin{proof}
Let $e \in E$ with $e \neq \emptyset$. We have:
$$
\bE_{X \sim \Omega^A}[\chi_e(X)] = 2^{-L} \sum_{y \in GF(2)^L} (-1)^{\sum_{i \in e} A(i) \cdot y} = 2^{-L} \sum_{y \in GF(2)^L} (-1)^{A(e) \cdot y} = 0
$$
\end{proof}

\subsection{Unbiased codes and fooling neighborhoods} We begin with an algorithm to construct a code which is  $E$-unbiased for a given set $E \subseteq 2^{[n]}$; we later extend this to fool neighborhoods.  This algorithm has two phases: first, we show how to find a code which is unbiased for \emph{most} of a given set $E$; we then bootstrap this to be unbiased on all of $E$. 

We begin with a simple result about multivariate polynomials over a finite field:
\begin{proposition}
\label{r1fnon-zero-poly1}
Let $p(z_1, \dots, z_k)$ be a non-zero polynomial over $GF(2^s)$, with degree at most $d$ in each variable separately. For any $\alpha \in GF(2^s)$, note that $p(\alpha, z_2, \dots, z_k)$ is a $k-1$ variable polynomial over $GF(2^s)$. If $\alpha \sim GF(2^s)$, then $p(\alpha, z_2, \dots, z_k) \equiv 0$ with probability at most $d/2^s$.
\end{proposition}
\begin{proof}
Factor $p$ as 
$$
p(z_1, \dots, z_k) = \sum_{t_2, \dots, t_k} q_{t_2,\dots, t_k}(z_1) z_2^{t_2} \dots z_k^{t_k},
$$
where $t_2, \dots, t_k$ range over non-negative integers.  Each such polynomial $q$ has degree $d$, and they are not all zero (else $p \equiv 0$). Let $t_2, \dots, t_k$ be such that $q_{t_2,\dots, t_k} \neq 0$; with probability at most $d/2^s$ we have $q_{t_2, \dots, t_k} (\alpha) = 0$. But if $q_{t_2, \dots, t_k}(\alpha) \neq 0$, then $p(\alpha, z_2, \dots, z_k)$ has a non-zero coefficient of $z_2^{t_2} \dots z_k^{t_k}$, hence $p(\alpha, z_2, \dots, z_k) \not \equiv 0$.
\end{proof}

\begin{proposition}
\label{r1fAchoose-prop2}
Let $E \subseteq 2^{[n]}$. Given integer parameters $k \geq 1, s \geq 0$, there is an algorithm to construct a code $A$ of length $s$, such that at most $k n^{1/k} 2^{-s} |E|$ sets $e \in E$ have $A(e) = 0$. This procedure has quasi-complexity $(k(s+ \log(mn)), 2^s W)$, where $m = |E|$ and $W = n + \sum_{e \in E} |e|$.
\end{proposition}
\begin{proof}
Let $Z$ be the set of formal monomials of the form $z_1^{u_1} \dots z_k^{u_k}$ in the ring $GF(2^s)[z_1, \dots, z_k]$, where $u_1, \dots, u_k \in \{0, \dots, d \}$ and $d = \lceil n^{1/k} - 1\rceil$. Enumerate $Z$ (in some arbitrary order) as $\mu_1, \dots, \mu_{\ell}$ where $\ell \geq n$.

For $i \in [n]$, we will define $A(i)$ to be the binary representation of $\mu_i(\alpha_1, \dots, \alpha_k)$, where $\alpha_1, \dots, \alpha_k$ will be chosen suitably from  $GF(2^s)$. For any $e \subseteq [n]$, define the polynomial $\mu_e = \sum_{i \in e} \mu_i$. By linearity, $A(e)$ is the binary representation of $\mu_e(\alpha_1, \dots, \alpha_k)$. So we need to select $\alpha_1, \dots, \alpha_k$ so that there are few sets $e \in E$ with $\mu_e(\alpha_1, \dots, \alpha_k) = 0$. We select $\alpha_1, \dots, \alpha_k$ sequentially, according to the following rule. For $i = 1, \dots, k+1$ let us define
$$
E_i = \{e \in E \mid \mu_e(\alpha_1, \alpha_2, \dots, \alpha_{i-1}, z_i, z_{i+1}, \dots, z_k) \not \equiv 0 \}
$$

By Proposition~\ref{r1fnon-zero-poly1}, if $\alpha_i$ is chosen uniformly at random, then in expectation at most $(d/2^s) |E_i|$ sets $e \in E_i$ satisfy $\mu_e(\alpha_1, \alpha_2, \dots, \alpha_i, z_{i+1}, \dots, z_k) \equiv 0$. By enumerating over all possible values of $\alpha_i \in GF(2^s)$ to maximize the size of $|E_{i+1}|$ we ensure that
$$
|E_{i+1}| \geq (1 - d/2^s) |E_i|
$$

As $E_1 = E$, at the end of this process we have $|E_{k+1}| \geq |E| (1 - d/2^s)^k \geq |E| (1 - k n^{1/k}/2^s)$, as required. This procedure requires $k$ separate stages. In each stage, we must count $|E_{i+1}|$ for every choice of $\alpha_i \in GF(2^s)$, which requires quasi-complexity $(s + \log mn, 2^s W)$.
\end{proof}

\begin{theorem}
\label{r1fheavy-thm1}
Let $E \subseteq 2^{[n]}$. There is an algorithm with quasi-complexity $(\log mn, W)$ to find an $E$-unbiased code of length $L = \log_2 m + O( \frac{\log m}{\log \log \log m})$, where $m = |E|$ and $W = n + \sum_{e \in E} |e|$.
\end{theorem}
\begin{proof}
We will first discuss our algorithm under the assumption that $m \geq n$.

We form the code $A$ by concatenating $r$ separate codes $A_1, \dots, A_r$, each of length $s$, i.e., 
$$
A(\ell) = \Bigl( A_1(\ell)(1), \dots, A_1(\ell)(s), A_2(\ell)(1), \dots, A_2(\ell)(s), \dots, A_r(\ell)(1), \dots, A_r(\ell)(s) \Bigr)
$$

The resulting code $A$ has length $L = r s$. We form $A_1, \dots, A_r$ sequentially, according to the following rule. For $i = 1, \dots, r+1$ define
$$
E_i = \{ e \in E \mid A_1(e) = \dots = A_{i-1}(e) = 0 \}
$$
Note that $E_1 = E$. Select each $A_i$ in turn by applying Proposition~\ref{r1fAchoose-prop2} to the set $E_i$, so that $|E_{i+1}| \leq \epsilon |E_i|$ for $\epsilon = k n^{1/k} 2^{-s}$, where we set $k = \lceil \log \log m \rceil, s = \lceil \frac{\log m}{\log \log \log m} \rceil$ and $r = \lceil 1 + \frac{\log m}{\log(1/\epsilon)} \rceil$. 

At the end of this process, we have $|E_{r+1}| \leq \epsilon^r |E| < 1$, and hence $E_{r+1} = \emptyset$, and hence the resulting code $A$ is $E$-unbiased. 

Using the fact that $n \leq m$, we may compute $r$ as 
$$
r \leq 2 + \frac{\log m}{\log(1/\epsilon)} \leq 2 - \frac{\log m}{\log \Bigl(2^{-\frac{\log m}{\log \log \log m}} (\log \log m + 1) n^{\frac{1}{\log \log m}}\Bigr) } \leq \frac{\log \log \log m}{\log 2} + O(1)
$$

So $L = r s \leq \log_2 m + O( \frac{\log m}{\log \log \log m})$. Overall, this procedure has quasi-complexity $(r k \log m, 2^s W) = (\log (mn), W)$.

Next, we discuss how to modify this procedure when $m < n$. In that case, with a simple pre-processing step of quasi-complexity $(\log mn, W)$, we can identify for each $e \in E$ a coordinate $v_e \in e$.  Let $V' = \{ v_e \mid e \in E \}$, and define $E' = \{e \cap V' \mid e \in E \}$. Using the above procedure we find a code $A'$ of length $L \leq  \log_2 m + O( \frac{\log m}{\log \log \log m})$ which is $E'$-unbiased. We finish by setting 
$$
A(\ell) = \begin{cases}
A'(\ell) & \text{if $\ell \in V'$} \\
0 & \text{otherwise}
\end{cases}
$$
\end{proof}

\subsection{Fooling neighborhoods}
If we wish to fool a list of neighborhoods $E \subseteq 2^{[n]}$, we could apply Theorem~\ref{r1fheavy-thm1} to the set $E' = \{f \mid f \subseteq e \in E \}$. However, even forming $E'$ directly might require exponential work. Instead, we can modify our algorithm to construct an $E'$-unbiased code, without needing to list $E'$ explicitly.

As in the proof of Proposition~\ref{r1fAchoose-prop2}, we associate to each $i \in [n]$ a distinct non-zero monomial $\mu_i$ over $GF(2^s)[z_1, \dots, z_k]$, wherein each indeterminate $z_i$ has degree at most $d = \lceil n^{1/k} - 1 \rceil$. We also define $\mu_e = \sum_{i \in e} \mu_i$ for any $e \subseteq [n]$. 

We will form the code $A$ as 
$$
A(\ell) = \Bigl( \mu_\ell(\alpha_{1,1}, \dots, \alpha_{1,k}), \mu_\ell(\alpha_{2,1},\dots, \alpha_{2,k}), \dots, \mu_\ell(\alpha_{r,1}, \dots, \alpha_{r,k}) \Bigr)
$$
for appropriate values $\alpha_{i,j} \in GF(2^s)$, where $i = 1, \dots, r$ and $j = 1, \dots, k$ (where here we identify elements of $GF(2^s)$ with binary vectors of length $s$). We also define $\alpha_{(i)} = (\alpha_{i,1}, \dots, \alpha_{i,k})$, and we write $\alpha$ as shorthand for $(\alpha_{1,1}, \dots, \alpha_{r,k})$.  

For $i = 1, \dots, r$ and $j = 1, \dots, k$ and $e \in E$, let us define the potential function
\begin{align*}
F_{i,j,e}(\alpha) = \sum_{\substack{f \subseteq e \\ f \neq \emptyset}} \Bigl[ \mu_f(\alpha_{(1)}) = \dots = \mu_f(\alpha_{(i-1)}) = 0 \wedge \mu_f(\alpha_{i,1}, \dots, \alpha_{i,j}, z_{j+1}, \dots, z_{k}) \equiv 0 \Bigr]
\end{align*}

The function $F_{i,j,e}(\alpha)$ only depends on $\alpha_{1, 1}, \dots, \alpha_{1,k}, \alpha_{2,1}, \dots, \alpha_{2,k}, \dots, \alpha_{i,1}, \dots, \alpha_{i,j}$. Note here that $\mu_f(\alpha_{(1)}), \dots, \mu_f(\alpha_{(i-1)})$ are elements of $GF(2^s)$ while $\mu_f(\alpha_{i,1}, \dots, \alpha_{i,j}, z_{j+1}, \dots, z_{k})$ is a polynomial in $k-j$ variables over $GF(2^s)$.  This function $F_{i,j,e}(\alpha)$ can be regarded as a type of pessimistic estimator for the number of sets $f \subseteq e$ for which the code $A$ will be biased.

\begin{proposition}
\label{eval-f-prop}
For any values $i,j,e$, the function $F_{i,j,e}$ can be computed with complexity $(\polylog(|e|, r, s), \poly(|e|,r,s))$.
\end{proposition}
\begin{proof}
Let $w = |e|$. If we associate the collection of subsets of $e$ with binary vectors of length $w$, then the set of all $f \subseteq e$ which satisfy the given constraint is a linear subspace $U$, and so $F_{i,j,e}(\alpha)$ has the value $(2^{\text{rank}(U)} - 1)$. Thus, we need to compute the rank of the set of vectors 
$$
X = \{ \mu_\ell(\alpha_{(1)}), \dots, \mu_{\ell}(\alpha_{(i)}), \mu_{\ell}(\alpha_{i,1}, \dots, \alpha_{i,j}, z_{j+1}, \dots, z_k) \mid \ell \in e \}.
$$ 
Here $\mu_{\ell}(\alpha_{i,1}, \dots, \alpha_{i,j}, z_{j+1}, \dots, z_k)$ is regarded as a listing of coefficients.

Let us count the length of each such vector. Each term $\mu_{\ell}(\alpha_{(t)})$ is an entry of $GF(2^s)$, hence has length $s$. Each value of $\ell$ corresponds to a distinct monomial $\mu_{\ell}(\alpha_{i,1}, \dots, \alpha_{i,j}, z_{j+1}, \dots, z_k)$, so over all we need to keep track of at most $w$ distinct monomials for the polynomial $\mu_{\ell}(\alpha_{i,1}, \dots, \alpha_{i,j}, z_{j+1}, \dots, z_k)$, for which each coefficient also has length $s$. In total,  the length of a vector $x \in X$ is at most $(r + w) s$, and the number of such vectors is $|X| = w$. There is an NC algorithm to compute matrix rank \cite{nc-gauss}; thus, this rank calculation has overall complexity $(\polylog(w, r, s), \poly(w,r,s))$.
\end{proof}

\begin{theorem}
\label{fool-nbhd-prop2}
Let $E \subseteq 2^{[n]}, m = |E|$ and $\text{width}(E) = w$. There is an algorithm with quasi-complexity $(w + \log(mn), (m+n) \poly(w))$  to produce a code $A$ of length $L \leq (1 + o(1)) (w + \log_2 m)$, such that $\Omega^A$ fools $E$.
\end{theorem}
\begin{proof}
We assume $n \leq m w$, as we can simply ignore all coordinates which do not appear in $E$. For $i = 1, \dots, r$ and $j = 1, \dots, k$ let us define
$$
H_{i,j}(\alpha)= \sum_{e \in E} F_{i,j,e}(\alpha)
$$

We also define $H_{0,k} (\alpha) = \sum_{e \in E} (2^{|e|} - 1)$, so for $i = 0, \dots, k$, we have 
$$
H_{i,k} (\alpha) = \sum_{e \in E} \sum_{\substack{f \subseteq e \\ f \neq \emptyset}} [\mu_f(\alpha_{(1)}) = \dots = \mu_f(\alpha_{(i)}) = 0]
$$

If $H_{r,k}(\alpha) = 0$, then the code $A$ is unbiased for every $f \subseteq e \in E$. Our strategy will be to loop over $i = 1, \dots, r$ and then $j = 1, \dots, k$, selecting $\alpha_{i,j}$ at each stage to minimize $H_{i,j}(\alpha)$.

We now make a few observations on the sizes of $F_{i,j,e}(\alpha)$. First, $H_0(\alpha) \leq 2^w m$. Also, since $\mu_f(z_1, \dots, z_k) \not \equiv 0$, we always have $F_{i,0,e}(\alpha) = 0$. Next, observe that if $\alpha_{i,j} \sim GF(2^s)$, then for any $f$ with $\mu_f(\alpha_{i,1}, \dots, \alpha_{i,j-1}, z_j, z_{j+1}, \dots, z_{k}) \not \equiv 0$, Proposition~\ref{r1fnon-zero-poly1} gives
$$
P( \mu_f(\alpha_{i,1}, \dots, \alpha_{i,j}, z_{j+1}, \dots, z_{k}) \equiv 0)  \leq d/2^s \leq n^{1/k}/2^s
$$

Consequently, when $\alpha_{i,j} \sim GF(2^s)$ and we condition on $\alpha_{1,1}, \dots, \alpha_{i,1}, \dots, \alpha_{i,j-1}$, we have
$$
\bE[ F_{i,j,e}(\alpha) ] \leq F_{i,j-1,e}(\alpha) + (n^{1/k}/2^s) (F_{i-1,k,e}(\alpha) - F_{i,j-1,e}(\alpha))
$$

By selecting $\alpha_{i,j}$ to minimize $H_{i,j}(\alpha)$, we thus ensure that
$$
H_{i,j}(\alpha) \leq H_{i,j-1}(\alpha) + (n^{1/k}/2^s) (H_{i-1,k}(\alpha) - H_{i,j-1}(\alpha)) \leq H_{i,j-1}(\alpha) + (n^{1/k}/2^s) H_{i-1,k}(\alpha)
$$

Since $F_{i,0,e} (\alpha) = 0$, this in turn ensures that $H_{i,k}(\alpha) \leq (k n^{1/k}/2^s) H_{i-1,k}(\alpha)$, so $H_{r, k}(\alpha) < (2^w m) (k n^{1/k} /2^s)^r$. Thus, for 
$$
r = \Big \lceil \frac{\log(2^w m)}{\log( 2^s/(k n^{1/k}))} \Big \rceil,
$$
the code $A$ will fool all of $E$. 

Now set $k = \lceil \log \log (m w) \rceil$ and $s = \lceil \frac{\log (m w)}{\log \log \log (m w)} \rceil$. Using the fact that $n \leq m w$, calculations similar to Theorem~\ref{r1fheavy-thm1} show that $r \leq (1+o(1)) (w + \log_2 m) \log \log \log(w m)/\log(w m)$. The code $A$ has length $L = r s \leq (1+o(1)) (w + \log_2 m)$.

Next let us examine the complexity of this process. In each iteration, we must evaluate $F_{i,j,e}$ for every $e \in E$ and every $\alpha_{i,j} \in GF(2^s)$. By Proposition~\ref{eval-f-prop}, each evaluation of $F_{i,j,e}$ has complexity $(\polylog(w, \log mn), \poly(w, \log mn))$.  Over all possible values $\alpha_{i,j} \in GF(2^s)$ and $e \in E$, this gives a total complexity of $( \polylog(w, \log mn), 2^s m \poly(w, \log mn))$.

There are $r k \leq \tilde O( \frac{w + \log m}{\log (mw)})$ iterations, so the overall complexity of this process is $\tilde O( \frac{w + \log m}{\log (mw)}) \times \tilde O(\log (mn) + (\log w)^{O(1)} )$ time and $(m + n)^{1+o(1)} w^{O(1)}$ processors. As $n \leq m w$, this simplifies to $\tilde O(w + \log(mn))$ time. 
\end{proof}

\subsection{Comparison with previous algorithms} Let us briefly compare Theorem~\ref{fool-nbhd-prop2} with previous algorithms for fooling neighborhoods. The simplest approach to fool a list $E \subseteq 2^{[n]}$, is to select a code $A$ whose dual code has weight $w+1$, where $w = \text{width}(E)$. The resulting probability space $\Omega^A$ is then $w$-wise-independent. There are algebraic constructions to do so efficiently; for example, \cite{alon-babai} discussed how to use BCH codes in this context for derandomization. Such codes have length roughly $(w/2) \log_2 n$.

An algorithm of Schulman \cite{schulman} reduces the code-length significantly to $O(w + \log |E|)$. To do so, it generates the set $E' = \{f \subseteq e \mid e \in E \}$ and then uses an algorithm fooling Fourier characters, similar to Theorem~\ref{r1fheavy-thm1}, to generate a code which is unbiased for $E'$. A similar approach is used in \cite{crs}, which interleaves other algorithmic steps with the generation of the code. The basic idea of both these works is to form $\Omega$ as a product of many independent copies of an $\epsilon$-approximately-independent probability space, where $\epsilon$ is constant. In \cite{schulman}, the underlying $\epsilon$-approximately-independent probability space was based on a construction of \cite{naor2} using Reed-Solomon codes; these have a particularly nice form for derandomizing part of the random seed.

These algorithms have high processor complexity (approximately $O(m n 2^w$)), and there are two main reasons for this. First, simply enumerating the set $E'$ requires a large processor count, exponential in $w$. Second, these algorithms test all possible seeds for the underlying Reed-Solomon code, and this requires a processor complexity exponential in the seed-length of that code. 

Theorem~\ref{fool-nbhd-prop2} thus improves in two ways over the previous algorithms. First, it has reduced time and processor complexity; in particular, it answers an open problem posed by Schulman \cite{schulman} in giving an $\text{NC}^1$ algorithm for $w \leq O(\log(mn))$, and it gives an NC algorithm for $w = \polylog(mn)$. Second, the code size is smaller: it gives $L \leq (1+o(1)) (w + \log_2 m)$ whereas the previous algorithms only guarantee $L \leq O(w + \log m)$.

\section{Fooling sums of juntas}
\label{r1fmain-alg-sec}
We say that a function $f: \{0, 1 \}^n \rightarrow \mathbf R$ is a \emph{$w$-junta} if there exists a set $Y = \{y_1, \dots, y_w \} \subseteq [n]$, such that
$$
f(x_1, \dots, x_n) = f'(x_{y_1}, \dots, x_{y_w})
$$
for some function $f': \{0, 1\}^w \rightarrow \mathbf R$.

In this section, we consider a function $S: \{0, 1 \}^n \rightarrow \mathbf R$ of the form
\begin{equation}
\label{dde}
S(x) = \sum_{j=1}^m f_j(x)
\end{equation}
where each $f_j$ is a $w$-junta whose value is determined by a variable subset $Y_j \subseteq [n]$. Our goal is to find some $x \in \{0, 1 \}^n$ with the property that $S(x) \geq \bE_{X \sim \{0, 1 \}^n} [S(X)]$. Our algorithm has four main components, which we will describe in turn:
\begin{enumerate}
\item We show how to apply conditional expectations when the objective functions are Fourier characters.
\item We apply Fourier decomposition to the sum (\ref{dde}), thus reducing a sum of $w$-juntas to a sum of Fourier characters. As we have discussed earlier, this step implemented directly has an exponential processor dependence on $w$.
\item We use (derandomized) random variable partitioning to break the overall sum into $w/w'$ subproblems involving $w'$-juntas, where $w' = o(\log n)$.
\item We introduce an object we refer to as \emph{partial-expectations oracle (PEO)}, which allows us to use conditional expectations to solve these subproblems sequentially. This requires $O(w)$ time, but only $2^{w'} = n^{o(1)}$ processors.
\end{enumerate}

Berger \& Rompel \cite{berger-simulating} discusses a few alternate strategies to mitigate the exponential dependence on $w$, for example when the underlying variables are drawn from $\{0, 1\}^b$ for $b = \polylog(n)$, or when $f_j$ are indicators of affine functions. But these strategies are not as general as we need for many applications. Our overall algorithm will handle all these situations as special cases.

\subsection{Conditional expectations for sums of Fourier characters and sums of $w$-juntas}
Our approach begins with a subroutine for optimization problems involving sums of Fourier characters. This idea has been used in a number of deterministic algorithms, starting with \cite{luby-old} and more extensively developed in \cite{berger-simulating, crs, harris2}. We present a slightly optimized form.

\begin{theorem}
\label{r1flinear-ce-thm}
Let $E \subseteq 2^{[n]}$ be given along with associated weights $\gamma_e$ for every $e \in E$. There is an algorithm with quasi-complexity $(\log mn, W)$ to find $x \in \{0, 1 \}^n$ such that
$$
\sum_{e \in E} \gamma_e \chi_e(x) \geq \gamma_{\emptyset},
$$
where $m = |E|$ and $W = n + \sum_{e \in E} |e|$.
\end{theorem}
\begin{proof} 
First, use Theorem~\ref{r1fheavy-thm1} to construct the $E$-unbiased code $A$ of length $L = O(\log m)$, using quasi-complexity $(\log mn, W)$. For any $y \in GF(2)^L$, let us define
$$
G(y) = \sum_{e \in E} \gamma_e (-1)^{A(e) \cdot y}
$$

We want to find $y \in GF(2)^L$ with $G(y) \geq \gamma_{\emptyset}$; we can then produce the desired $x \in \{0, 1 \}^n$ by setting $x_i = A(i) \cdot y$ for $i = 1, \dots, n$. Since $A$ is $E$-unbiased, we have $\bE_{y \sim GF(2)^L}[ (-1)^{A(e) \cdot y}] = 0$ for every $e \neq \emptyset$. So $\bE_{y \sim GF(2)^L}[G(y)] = \gamma_{\emptyset}$, and thus a satisfying $y \in GF(2)^L$ exists.

To find it, we use conditional expectation: we guess chunks of $t = \frac{\log mn}{\log \log mn}$ bits of $y$ at a time, to ensure that the expected value of $\bE[G(y)]$ increases. For each such guess, we will compute in parallel the resulting expected value $\bE[G(y)]$, when certain bits of $y$ are fixed and the rest remain independent fair coins. We may compute the conditional expectations of a term $(-1)^{A(e) \cdot y}$, using the following observation: suppose that $y_1, \dots, y_k$ are determined while $y_{k+1}, \dots, y_L$ remain independent fair coins. Then $\bE[(-1)^{A(e) \cdot y}] = 0$ if $A(i) = 1$ for any $i \in e \cap \{k+1, \dots, L \}$, and otherwise $(-1)^{A(e) \cdot y}$ is determined by $y_1, \dots, y_k$.

This process requires  $\lceil L/t \rceil \leq O( \log \log mn )$ rounds. For each possible value for a $t$-bit chunk of $y$, evaluating $G$ has complexity $(\log(mn), W^{1+o(1)})$. As $t \leq o(\log mn)$, we get an overall quasi-complexity of $\tilde O(\log mn, W)$.
\end{proof}

\begin{lemma}
\label{junta-thm1}
Suppose we have a full listing of the truth-table of each $f_j$. There is an algorithm to find $x \in \mathcal \{0, 1 \}^n$ satisfying $S(x) \geq \bE_{X \sim \{0, 1 \}^n} [S(X)]$, using quasi-complexity $(w + \log mn, 2^w m + n)$.
\end{lemma}
\begin{proof}
Using Proposition~\ref{r1ffourier-prop}, transform each $f_j$ as $f_j(x) = \sum_{e \subseteq Y_j} \gamma_{j,e} \chi_e (x)$. This step has quasi-complexity $(w + \log mn, 2^w m)$. We thus have:
$$
S(x) = \sum_j \sum_{e \subseteq Y_j} \gamma_{j,e} \chi_e(x) = \sum_e \chi_e(x) (\sum_j \gamma_{j, e})
$$
and $\sum_j \gamma_{j,\emptyset} = \bE_{X \sim \{0, 1 \}^n} [S(X)]$.

Next apply Theorem~\ref{r1flinear-ce-thm} to the set $E = \{e \mid e \subseteq Y_j \}$ and associated weights $\sum_j \gamma_{j,e}$. Since $|E| \leq 2^w m$ and $\text{width}(E) \leq w$, this procedure has quasi-complexity $(\log (2^w m n), 2^w m + n)$. 
\end{proof}

When $w = \polylog(n)$, this means that Lemma~\ref{junta-thm1} gives quasi-NC algorithms. When $w = \Theta(\log n)$ then Lemma~\ref{junta-thm1} gives NC algorithms; however, the processor complexity (while polynomial) may be quite large, depending on the size of $w$. 

As a side application of Theorem~\ref{r1flinear-ce-thm} (which is not needed for our overall derandomization approach), let us consider the \emph{heavy-codeword problem.} We are given a code $A$ of length $L$ and size $n$, presented as a $L \times n$ generator matrix. Our goal is to find a codeword whose weight is at least the expected weight of a randomly-chosen codeword.  This was  introduced as a toy derandomization problem by \cite{naor2}; this work also gave an algorithm with complexity roughly $(\log Ln, L^2 n^2)$. This was later improved by \cite{crs} to complexity $(\log Ln , L n^2)$. We improve this further to nearly optimal time and processor complexities.
\begin{corollary}
\label{r1heavy-corr2}
There is an algorithm with quasi-complexity $(\log Ln, Ln)$ to find a heavy codeword.
\end{corollary}
\begin{proof}
We suppose without loss of generality that no row of $A$ is all zero. In this case, the expected weight of a codeword is $L/2$. Letting $y_1, \dots, y_L$ denote the rows of $A$, we wish to find a vector $x \in \{0, 1 \}^n$ such that $y_j \cdot x = 1$ for at least $L/2$ values of $j$. 

Define $S(x) = -\sum_{j=1}^L \chi_{y_j}(x)$. If $S(x) \geq 0$ then $x$ is orthogonal to at least half of $y_1, \dots, y_L$ as desired. So we apply Theorem~\ref{r1flinear-ce-thm}, noting that $W \leq Ln$. Since $y_1, \dots, y_L$ are all distinct from zero, we have $\gamma_{\emptyset} = 0$.
\end{proof}

\subsection{Derandomized variable partitioning.}
This step is based on a derandomization technique of \cite{alon-srin} using symmetric polynomials and approximately-independent probability spaces (also known as small-bias probability spaces). We begin by defining and quoting some results on approximate independent probability spaces.
\begin{definition}
\label{approx-def}
A probability space $\Omega$ over $\{0, 1 \}^n$ is $t$-wise, $\epsilon$-approximately independent, if for any indices $1 \leq i_1 < i_2 < \dots <  i_t \leq n$, and any bits $y_1, \dots, y_t \in \{0, 1 \}^t$, we have
$$
P_{x \sim \Omega} (x_{i_1} = y_1 \wedge \dots \wedge x_{i_t} = y_t) \leq (1+\epsilon) 2^{-t}
$$
\end{definition}

\begin{theorem}[\cite{naor2}]
\label{naor2thm}
For any integer $t \geq 1$ and $\epsilon > 0$, there is a $t$-wise, $\epsilon$-approximately independent probability space $\Omega$ of support size $|\Omega| \leq 2^{O(t + \log (1/\epsilon) + \log \log n)}$. The space $\Omega$ can be constructed with quasi-complexity $(t + \log(1/\epsilon) + \log n, 2^{O(t + \log (1/\epsilon) + \log \log n)})$
\end{theorem}

\begin{lemma}
\label{r1fdiscrepthm}
Let $E \subseteq 2^{[n]}$, where $m = |E|$ and $w = \text{width}(E)$. One can construct a partition of $[n]$ into $R$ parts  $[n] = T_1 \sqcup T_2 \sqcup \dots \sqcup T_R$, for some $R = O( 1+ \frac{w (\log \log mn)^5}{ \log mn})$, satisfying
$$
|f \cap T_k| \leq O(\frac{\log mn}{\log \log \log mn}) \text{ for all $f \in E, k \in [R]$}
$$

This algorithm has quasi-complexity $(\log w \log (mn), w^{O(1)} (m+n))$.
\end{lemma}

\begin{proof}

Let $r = \lceil \log_2 \frac{C w (\log \log mn)^5}{\log mn} \rceil$, where $C$ is a constant to be specified. We will construct binary vectors $y_1, \dots, y_r \in \{0, 1 \}^n$ and then define for each $k \in \{0,1 \}^r, \ell \in \{0, \dots, r\}$ the sets $T_k^{\ell} \subseteq[n]$ by
$$
T_k^{\ell} =  \{i \in [n] \mid y_1(i) = k(1) \wedge y_2(i) = k(2) \wedge \dots \wedge y_{\ell}(i) = k(\ell)  \}
$$

We will finish by setting $R = 2^r$ and forming the sets $T_1, \dots, T_R$ by $T_k = T_k^r$ where $k$ ranges over $\{0, 1 \}^r$. For each $k \in \{0,1 \}^r, \ell \in \{0, \dots, r \}, f \in E$ we define $H_\ell(f,k) = |f \cap T_k^{\ell}|$. We will achieve the goal of the theorem if we select $y_1, \dots, y_r$ so that every $f \in E, k \in \{0, 1 \}^r$ has $H_r(f,k) \leq t$ for $t = \lceil \frac{\log m n}{\log \log \log m n} \rceil$.

For each $\ell = 0, \dots, r$ let us define the potential function
$$
Q_{\ell} = \sum_{f,k} \binom{H_{\ell}(f,k)}{t}
$$
Observe that $Q_r$ is an integer; thus, if $Q_r < 1$, then it follows that $Q_r = 0$ and so $H_r(f,k) < t$ for all $f,k$ as desired. 

Let $\Omega$ be a probability distribution over $GF(2)$ which is $t$-wise, $\epsilon$-approximately independent, where $\epsilon = 1/r$, according to Definition~\ref{approx-def}. By Theorem~\ref{naor2thm}, we have $|\Omega| \leq (mn)^{o(1)} w^{O(1)}$; furthermore, the complexity of generating $\Omega$ will be negligible for the overall algorithm.

For $y_{\ell} \sim \Omega$, each $t$-tuple of elements in $f \cap T_k^\ell$ has a probability of at most $2^{-t} (1+\epsilon)$ of surviving to $T_k^{\ell+1}$. This implies that
$$
\bE[ \tbinom{H_{\ell+1}(f,k)}{t} \mid y_1, \dots, y_{\ell} ] \leq (1+\epsilon) 2^{-t} \tbinom{ H_{\ell}(f,k)}{t}
$$
and consequently $\bE[Q_{\ell+1} \mid y_1, \dots, y_{\ell}] \leq (1+\epsilon)  2^{-t}  Q_{\ell}$.

Our algorithm is to select $y_1, \dots, y_r$ sequentially in order to minimize $Q_{\ell+1}$ at each stage $\ell$. This ensures that $Q_{\ell} \leq (1+\epsilon)  2^{-t} Q_{\ell-1}$, and so at the end of the process we have
{\allowdisplaybreaks
\begin{align*}
Q_r &\leq (1+\epsilon)^r 2^{-t r} Q_0 = (1 + 1/r)^r R^{-t} \sum_{f,k} \binom{H_0(f,k)}{t} \leq e R^{1-t} m \binom{w}{t} \leq e R^{1-t} m (e w/t)^t
\end{align*}
}

Simple calculations now show that $Q_r < 1$ for $C$ a sufficiently large constant.

We now examine the complexity of this algorithm. There are $r$ stages; in each stage, we must search the probability space $\Omega$ and compute $Q_{\ell}$. The potential function $Q_{\ell}$ can be computed with quasi-complexity $(\log mn, m w R)$. Note now that $m w R \leq (mn)^{o(1)} m w^{O(1)}$. As $|\Omega|  \leq (mn)^{o(1)} w^{O(1)}$, this costs $(mn)^{o(1)} w^{O(1)} (m + n)$ processors and $\tilde O( r \log mn) = \tilde O(\log w \log mn)$ time.
\end{proof}

\subsection{The partial-expectations oracle}
As we have discussed, we need implicit access to $f_j$ in order to avoid the exponential dependence on $w$. A key idea of Berger \& Rompel \cite{berger-simulating} to achieve this is an algorithm capable of determining certain conditional expectations for the objective functions.
\begin{definition}
Algorithm $\mathcal A$ is a \emph{partial-expectations oracle (PEO)} for the functions $f_j$, if it is capable of the following operation. Given any $X' \in \{0, 1, \text{?} \}^n$, the algorithm $\mathcal A$ computes $F_j = \bE_{X \sim \Omega}[f_j(X)]$ for $j = 1, \dots, m$, where the probability distribution $\Omega$ is defined by drawing each bit $X_i$ independently, such that if $X'_i = \text{?}$ then $X_i$ is Bernoulli-$1/2$ and if $X'_i \neq \text{?}$ then $X_i = X'_i$.
\end{definition}

We note that this form of PEO is simpler than that used by Berger \& Rompel: the latter requires evaluating the conditional expectation of $f_j(X)$ given that $X$ is confined to an affine subspace, while our PEO only requires computing this conditional expectation when individual bits of $X$ are fixed. 

We now combine all the ingredients to obtain our conditional expectations algorithm.
\begin{theorem}
\label{r1fpolylog-thm}
Suppose $S(x) = \sum_{j=1}^m f_j(x)$ for $x \in \{0, 1\}^n$, where each $f_j$ is a $w$-junta. Suppose we have a PEO for the functions $f_j$ with complexity $(C_1, C_2)$. 

Then there is an algorithm to find a vector $x$ satisfying 
$$
S(x) \geq \bE_{X \sim \{0,1 \}^n} S(X),
$$ 
with quasi-complexity $( C_1 (1 +\frac{w}{\log mn}), w^{O(1)} C_2)$.
\end{theorem}
\begin{proof}
We assume $C_1 \geq \Omega(\log mn)$ and $C_2 \geq \Omega(m+n)$ as it requires this complexity to take as input the values $j, X'$ and output $F_j$. We similarly assume that $n \leq m w$, as variables not involved in any objective function may be ignored.

First apply Lemma~\ref{r1fdiscrepthm} to determine a partition $[n] = T_1 \sqcup \dots \sqcup T_R$ for some $R = O( 1+ \frac{w (\log \log mn)^5}{ \log mn})$, such that $|Y_j \cap T_k| \leq w'$ for some $w' \leq O( \frac{\log mn}{\log \log \log mn} )$.  This stage has quasi-complexity $(\log w \log mn, \penalty 0 w^{O(1)} (m + n))$.

Next, for $r = 1, \dots, R$, we seek to determine the bits $\{x_i  \mid i \in T_r \}$. Define the function $f'_j(z)$ to be the expected value of $f_j(X)$, when the entries $X_i$ for $i \in T_r$ are set to $z_i$, the variables $X_i$ for $i \in T_1, \dots, T_{r-1}$ are set to $x_i$, and the remaining entries of $X$ (for $i \in T_{r+1}, \dots, T_R$) remain fair coins. Each $f'_j$ is a $w'$-junta and we can determine its truth-table $f'_j$ using $2^{w'}$ invocations of our PEO, where we define $X'_i = \text{?}$ for $i \in T_{r+1} \cup \dots \cup T_{R}$ and $X'_i \neq \text{?}$ otherwise. This in turn requires $C_1 + \tilde O(\log(2^{w'} mn)) \leq \tilde O(C_1)$ time and $(mn)^{o(1)} 2^{w'} C_2 \leq (mn)^{o(1)} C_2$ processors.

Next, apply Lemma~\ref{junta-thm1} to determine a value for the relevant variables in $T_r$; this step takes $\tilde O(w' + \log mn) \leq \tilde O(\log mn)$ time and $(n + 2^{w'} m)^{1+o(1)} \leq (mn)^{o(1)} (m+n)$ processors.

Over all $R$ stages, the total time for this algorithm is $\tilde O(R C_1) \leq \tilde O(\frac{w C_1}{\log mn} + w + C_1)$.
\end{proof}

We emphasize the low time and processor complexity of this algorithm. For example, if $w = \polylog(mn)$ and $C_1 = \tilde O(\log mn)$ (which are typical parameters), then this has quasi-complexity $(w, C_2)$. Even if $w = \Theta(\log mn)$, this can lead to greatly reduced complexities as compared to the algorithm of \cite{berger-simulating}.

This algorithm requires an appropriate PEO, which must be constructed in a problem-specific way. One important class of objective functions, which was one of the main cases considered by Berger \& Rompel \cite{berger-simulating}, is indicator functions for affine spaces; PEO's for such functions can be derived by a rank calculation. We will consider more complicated types of PEO's; one significant difficulty, as we discuss next, is that many objective functions are naturally represented as functions of integer-valued variables (not just isolated bits).

\subsection{Non-binary variables}
\label{r1multi-bit-sec}
Let us consider a slightly more general setting: we have $n$ variables $x_1, \dots, x_n$, each of which is an integer in the range $\{0, \dots, 2^b - 1 \}$. Our objective function is again a sum of $w$-juntas, that is, each $f_j(x)$ depends on at most $w$ coordinates of $x$.  This can easily be reduced to the model we have discussed earlier: we replace each variable $x_i$ with $b$ separate binary variables $x_{i1}, \dots, x_{ib}$. Now each $f_j$ depends on $w b$ bits of the expanded input, and so is a $wb$-junta.

However, there is a complication. In order to apply Theorem~\ref{r1fpolylog-thm}, we need a PEO for the functions $f_j$. Thus we need to compute the expected value of $f_j(x)$, given that certain \emph{bits} of $x$ are fixed to specific values. This can be somewhat awkward, as restricting arbitrary bits of $x_i$ does not necessarily have any natural interpretation when $x_i$ is an integer in the range $\{0, \dots, 2^b - 1 \}$. It is often easier to use the strategy of \cite{berger-simulating}, which fixes the bit-levels of $x_1, \dots, x_n$ one at a time. This allows us to use a simpler type of PEO where the pattern of known/unknown bits is more controlled.

For the purposes of the algorithm, we identify the integer set $\{0, \dots, 2^b - 1 \}$ with the set of length-$b$ binary vectors; a vector $(x_0, \dots, x_{b-1})$ corresponds to the integer $\sum_{i=0}^{b-1} 2^i x_{b-1-i}$. Note here that $x_0$ is the \emph{most-significant bit}. Let us define $\mathcal M_b$ as the set $\{0, \dots, 2^b - 1 \}$ equipped with this bit-based interpretation. Likewise, if $x \in \mathcal M_b^n$, we let $x(i, j)$ denote the $j^\text{th}$ most significant bit of the integer value $x_i$.

\begin{definition}
\label{grade-def}
We say that $X' \in \{0, 1, \text{?} \}^{n b}$ is \emph{graded} if there is some integer $\ell \in \{0, \dots, b - 1 \}$ such that for all $i, j$ the following two conditions hold:
\begin{enumerate}
\item $X'(i, j) \in \{0, 1 \}$ for $j = 0, \dots, \ell-1$
\item $X'(i,j) = \text{?}$ for $j = \ell+1, \dots, b-1$
\end{enumerate}

We say that $X'$ is \emph{fully-graded} if $X'$ satisfies for some integer $\ell \in \{0, \dots, b \}$ the stricter condition that for all $i, j$ the following two conditions hold:
\begin{enumerate}
\item $X'(i,j) \in \{0, 1 \}$ for $j = 0, \dots, \ell-1$
\item $X'(i,j) = \text{?}$ for $j = \ell,  \dots, b-1$
\end{enumerate}

An algorithm $\mathcal A$ is a \emph{graded PEO} (respectively \emph{fully-graded PEO}) for the functions $f_j$ if it is a PEO, but only for queries $X'$ which are graded (respectively, fully-graded).
\end{definition}

\begin{theorem}
\label{r1fpolylog-thm2}
Suppose that $S(x) = \sum_{j=1}^m f_j(x)$ for $x \in \mathcal M_b^n$, where each function $f_j$ is a $w$-junta, and we have a graded PEO for the functions $f_j$ with complexity $(C_1, C_2)$.

Then we can find $x \in \{0, 1 \}^n$ satisfying $S(x) \geq \bE_{X \sim \mathcal M_b^n}[S(X)]$, using quasi-complexity $(b C_1 (1 + \frac{w}{\log(mn)}), w^{O(1)} C_2)$.
\end{theorem}
\begin{proof}
We will determine the bits of $x$ in $b$ separate stages; at the $\ell^{\text{th}}$ stage, we determine the bit-level $\ell$ of each entry $x_i$. For $\ell = 0, \dots, b-1$, consider the following process. Define the function $f_{\ell,j}(z)$ to be the expected value of $f_{j}(X)$, when the bit-levels $0, \dots, \ell-1$ of $X$ are taken from the already-determined vector $x$; when the bit-level $\ell$ of $X$ is set to $z$; and when the bit-levels $\ell+1, \dots, b-1$ of $X$ are independent fair coins.

Each $f_{\ell,j}$ is a $w$-junta, and the graded PEO for the functions $f_j$ yields a PEO for the functions $f_{\ell,j}$. Therefore Theorem~\ref{r1fpolylog-thm} produces a $z \in \{0, 1 \}^n$ with $\sum_j f_{\ell, j}(z) \geq \bE_{Z \sim \{0, 1 \}^n } [ \sum_j f_{\ell,z}(Z)]$.
\end{proof}

One important application of non-binary variables concerns derandomizing biased coins. For a vector of probabilities $p \in [0,1 ]^n$, consider the probability space  with $n$ independent variables $X_1, \dots, X_n$, wherein each $X_i$ is Bernoulli-$p_i$. We write this more compactly as $X \sim p$. Most of our derandomization results we have proved earlier have assumed that the underlying random bits are independent fair coins (i.e. with $p_1, \dots, p_n = 1/2$).

\begin{definition}
An algorithm $\mathcal A$ is a \emph{continuous PEO} for the functions $f_j$, if it is capable of the following operation. Given any vector $q \in [0,1]^n$, whose entries are rational number with denominator $2^b$, the algorithm $\mathcal A$ computes $F_j = \bE_{X \sim q}[f_j(X)]$ for $j = 1, \dots, m$.
\end{definition}

Note that a PEO can be regarded as a special case of a continuous PEO, in which the probability vector $q$ is restricted to the entries $\{0, 1/2, 1 \}$.

\begin{theorem}
  Suppose that $S(x) = \sum_{j=1}^m f_j(x)$ for $x \in \{0, 1\}^n$, where each function $f_j$ is a $w$-junta, and we have a continuous PEO for the functions $f_j$ with complexity $(C_1, C_2)$.

  Let $p \in [0,1]^n$ be a vector of probabilities, wherein each entry $p_i$ is a rational number with denominator $2^b$. Then we can find a vector $x$ satisfying $S(x) \geq \bE_{X \sim p}[S(X)]$, using quasi-complexity  $(b C_1 (1 + \frac{w}{\log mn}), (w b)^{O(1)} C_2))$.
\end{theorem}
\begin{proof}
Consider the function $f'_j: \mathcal M_b^n \rightarrow \mathbf R$ defined by 
$$
f'_j(y_1, \dots, y_n) = f_j \bigl( [y_1/2^b \leq p_1], \dots, [y_n/2^b \leq p_n] \bigr)
$$

Each function $f'_j$ depends on $w$ coordinates of $y$. Furthermore, if certain most-significant bit levels of $y$ are fixed to a certain value and the remaining least-significant bit-levels of $y$ are independent fair coins, then each term $[y_i/2^b \leq p_i]$ is a Bernoulli-$q_i$ variable, where $q_i$ depends on the fixed values of $y_i$. Therefore, the given continuous PEO for $f_j$ provides a graded PEO for the functions $f'_j$, with a complexity of $(\log (nb) + C_1, C_2 + nb)$. 

Finally, observe that when $Y \sim \mathcal M_b^n$, each term $[Y_i/2^b \leq p_i]$ is Bernoulli-$p_i$; therefore, we have
$$
\bE_{Y \sim \mathcal M_b^n} [ \sum_j f'_j(Y) ]  = \bE_{X \sim p} [ \sum_j f_j(X)]
$$

So Theorem~\ref{r1fpolylog-thm2} produces $y_1, \dots, y_n \in \mathcal M_b^n$ with $\sum_j f'_j(y) \geq \bE_{X \sim p} [ \sum_j f_j(X)]$. Output the vector $x \in \{0,1 \}^n$ defined by $x_i = [y_i/2^b \leq p_i ]$.
\end{proof}

This leads to PEOs for the class of functions computed by a read-once branching program (ROBP). In this computational model, the function $f$ is represented as a directed acyclic graph; at each node $v$, a single variable $x_v$ is read and the program branches to two possible destinations depending on the variable $x_v$. There is a designated starting vertex and at some designated sink vertices, a real number is output. In addition, every variable label appears at most once on each directed path. This is a quite general class of functions, which includes log-space statistical tests as used by Sivakumar's derandomization \cite{sivakumar}. See \cite{borodin-cook} for further details.
\begin{proposition}
\label{robp2}
If a $w$-junta $f$ can be computed by a ROBP on $M$ states, then it has a continuous PEO with quasi-complexity $(\log b \log w \log Mw, b M^3 w^{O(1)})$.
\end{proposition}
\begin{proof}
We must calculate the expected value of $f$, given that the variables $X_1, \dots, X_n$ are independent Bernoulli-$q_i$. Now observe that, for any states $s_1, s_2$, the probability that $s_1$ goes to $s_2$ in at most $h$ time-steps
is the sum over intermediate states $s$ of the probability that $s_1$ goes to $s$ in at most $h/2$ time-steps and that $s$ goes to $s_2$ in at most $h/2$ time-steps; this follows from the definition of an ROBP.  Using this relation, one may recursively build the transition matrix for pairs of states $s_1 \rightarrow s_2$ over time horizons $h = 1, 2, 4, \dots, w$. Each such iteration takes time $\tilde O(\log b \log Mn)$ and there are $\tilde O(\log w)$ iterations.
\end{proof}

\subsection{Application: rainbow hypergraph coloring}
\label{r1frainbow-sec}
As a simple example of our derandomization method, let us consider a $d$-uniform  hypergraph $H$, with $m$ edges and $n$ vertices. We say that an edge $e$ is \emph{rainbow} for a vertex coloring of $H$, if all its vertices receive distinct colors. A challenge determinization problem is to find a $d$-coloring with $m \frac{d!}{d^d}$ rainbow edges, which is the expected number in a uniform random  coloring. In \cite{alon-babai}, an NC algorithm was given in the case $d = O(1)$. This was extended by \cite{berger-simulating} to arbitrary $d$; although \cite{berger-simulating} did not give any concrete time or processor complexity, the complexity appears to be roughly $(\log^4 mn, n + m^{1 + \log 2})$. 

We significantly improve both the time and processor costs. Note that it requires $\Omega(m d + n)$ space to store the hypergraph.
\begin{theorem}
\label{r1frainbow-thm}
There is an algorithm to find a vertex coloring with $d$ colors and at least $m \frac{d!}{d^d}$ rainbow edges, using quasi-complexity $(\log^2 mn, m d + n)$.
\end{theorem}
\begin{proof}
We begin with simple pre-processing steps using complexity $(\log(mn), md + n)$. First, when $d \geq \log m + C \log \log m$ for a sufficiently large constant $C$, then it suffices to rainbow-color a single edge, which may be done easily. Second, when $n < m d$, then some vertex is not used; we may delete it from the graph. Hence we assume $m \geq n d$ and $d \leq \log m + O(\log \log m)$.

Given a binary vector $x \in \mathcal M_b^n$, we define the associated $d$-coloring $\phi: V \rightarrow \{0, \dots, d-1 \}$ by $\phi_x(v) = \lfloor (d/2^b) x_v \rfloor$. For each edge $e \in H$ let $f_e(x)$ be the indicator function that $e$ is rainbow on the coloring $\phi_x$, and define $S(x) = \sum_{e} f_e(x)$.

As shown in \cite{berger-simulating}, by taking $b = \tilde O(\log mn)$, we can ensure that $\bE_{X \sim \mathcal M_b^n} [S(X)] > \frac{(m d!-1)}{d^d}$. Furthermore, since $S(x)$ is an integer and $\frac{(m d!-1)}{d^d}$ is a rational number with denominator $d^d$, when $S(x) > \frac{(m d!-1)}{d^d}$ we ensure that $S(x) \geq \frac{m d!}{d^d}$. So, if we find $x \in \mathcal M_b^n$ with $S(x) \geq \bE_{X \sim \mathcal M_b^n} [S(X)]$, then this will yield our desired coloring.

Since each $f_e$ is a $d$-junta, we apply Theorem~\ref{r1fpolylog-thm2} with $w = d$ to construct $\phi_x$ using total quasi-complexity of $O( b C_1 (1 + \frac{d}{\log mn}), n +  C_2)$, where $(C_1, C_2)$ is the complexity of PEO for the functions $f_e$.   Observe that $b \leq \tilde O(\log mn), d \leq \log m + O(\log \log m)$. In Proposition~\ref{r1fgraded-rainbow-oracle} (which we defer to the appendix) we show that  $C_1 \leq \tilde O(\log mn)$ and $C_2 \leq (m+n)^{1+o(1)}$, so this simplifies to quasi-complexity of $(\log^2 mn, m+n)$.
\end{proof}

\section{The Lov\'{a}sz Local Lemma with complex bad-events}
\label{r1fmt-sec}
The Lov\'{a}sz Local Lemma (LLL) is a keystone principle in probability theory which asserts that if one has a probability space $\Omega$ and and a set $\mathcal B$ of ``bad-events'' in $\Omega$, then under appropriate ``local'' conditions there is a positive probability that no event in $\mathcal B$ occurs. The LLL has numerous applications to combinatorics, graph theory, routing, etc.   The simplest ``symmetric'' form of the LLL states that if each bad-event $B \in \mathcal B$ has probability $P_{\Omega}(B) \leq p$ and affects at most $d$ bad-events (including itself), then if $e p d \leq 1$ then $P( \bigcap_{B \in \mathcal B} \overline{B} ) > 0$.

Although the LLL applies to general probability spaces, in most applications a simpler bit-based form suffices, wherein the space $\Omega$ has $n$ variables $x_1, \dots, x_n$, which are independently drawn from $\mathcal M_b$. In this setting, each $B \in \mathcal B$ is a boolean function $f_B$ on a subset $Y_B$ of the variables. We say that bad-events $B, B'$ affect each other (and write $B \sim B'$) if $Y_B \cap Y_{B'} \neq \emptyset$. We say that $x \in \mathcal M_b^n$ \emph{avoids $\mathcal B$} if $f_B(x) = 0$ for all $B \in \mathcal B$. 

In a seminal paper \cite{mt}, Moser \& Tardos introduced the following simple randomized algorithm, which we refer to as the MT algorithm, giving efficient randomized constructions for nearly all LLL applications in bit-based probability spaces.

\begin{algorithm}[H]
\centering
\begin{algorithmic}[1]
\State Generate $x_1, \dots, x_n$ as independent fair coins.
\While{some bad-event is true on $x$}
\State Arbitrarily select some true bad-event $B$
\State For each $i \in Y_B$, draw $x_i$ as an independent fair coin. (We refer to this as \emph{resampling} $B$.)
\EndWhile
\end{algorithmic}
\caption{The Moser-Tardos algorithm}
\end{algorithm}

Under nearly the same conditions as the probabilistic LLL, the MT algorithm terminates in polynomial expected time. Moser \& Tardos also gave a parallel (RNC) variant of this algorithm, requiring a slack compared to the LLL criterion.

There are two key analytic techniques introduced by \cite{mt} for this algorithm. The first is the idea of a \emph{resampling table}. In the MT algorithm as we have presented it, the new values for each variable are drawn in an online fashion. Instead, one can imagine a fixed table $R$. This table records, for each variable $i$, an infinite list of values $R(i,1), R(i,2), \dots, $ for that variable, which are all independent draws from $\mathcal M_b$. When the MT algorithm begins, it sets $x_i = R(i,1)$ for each variable $i$; if a variable $x_i$ needs to be resampled, it sets $x_i = R(i,2)$, and so forth. Once we have fixed a resampling table $R$, the MT algorithm can be executed deterministically.

We view the resampling table $R$ as a function $R: [n] \times \mathbf Z_+ \rightarrow \mathcal M_b$. We define a \emph{slice} to be a set $W \subseteq [n] \times \mathbf Z_+$ with the property that each $i \in [n]$ has at most one $j \in \mathbf Z_+$ with $(i,j) \in W$. For such a slice $W$, sorted as $W = \{ (i_1, j_1), \dots, (i_k, j_k) \}$ with $i_1 < i_2 < \dots < i_k$, we define the projection $\pi_W$ by setting $\pi_W(R) = ( R(i_1, j_1), \dots, R(i_k, j_k) )$

The other key idea introduced by Moser \& Tardos is the \emph{witness tree}, which represents a possible execution path for the MT algorithm leading to a given resampling.  This is explained in great detail in \cite{mt}, which we recommend as an introduction. As a brief summary, suppose we want to explain why some bad-event $B$ was resampled at time $t$. We form a witness tree $\tau$ by first placing a root node labeled by $B$, and then going in time through the execution log from time $t-1$ to time $1$. For each event $B$ we encounter at time $s < t$, we look in $\tau$ to find if there is some node $v'$ labeled by $B' \sim B$. If so, we place a node $v$ labeled by $B$ in the tree as a child of $v'$; if there are multiple choices for $v'$, we always choose the one of greatest depth (if there are multiple choices at greatest depth, we break the tie arbitrarily). 

For any witness tree $\tau$ and any node $v \in \tau$, we let $L(v) \in \mathcal B$ denote the label of $v$.

\begin{definition}[Weight and size of witness tree]
For a witness tree $\tau$, we define the \emph{size} of $\tau$ as the number of nodes in $\tau$ and we define the \emph{weight} of $\tau$ as $w(\tau) = \prod_{v \in \tau} P_{\Omega}(L(v))$.
\end{definition}

The most important result of \cite{mt}, which explains why the MT algorithm works, is the \emph{Witness Tree Lemma}:

\begin{lemma}[\cite{mt}]
\label{r1fwitness-tree-lemma}
The probability that a witness tree $\tau$ appears during the execution of the MT algorithm is at most $w(\tau)$.
\end{lemma}

To prove this Lemma, \cite{mt} shows that $\tau$ imposes certain conditions on the resampling table $R$. 
\begin{lemma}[\cite{mt}]
\label{r1fwitness-tree-lemma2}
For any witness tree $\tau$ on $t$ nodes, there is a set of slices $W_v$, indexed by nodes $v \in \tau$, such that
\begin{enumerate}
\item[(A1)] For $v \neq v'$ we have $W_v \cap W_{v'} = \emptyset$.
\item[(A2)] A necessary condition for $\tau$ to appear is that $f_{L(v)}( \pi_{W_v}(R)) = 1$ for every $v \in \tau$.
\item[(A3)] The sets $W_v$ can be determined from $\tau$ with quasi-complexity $(\log nt, n t)$.
\item[(A4)] Every $v \in \tau$ has $|W_v| = |Y_{L(v)}|$.
\end{enumerate}
\end{lemma}
\begin{proof}
For each node $v \in \tau$ and each $i \in [n]$, let $u_{i,v}$ denote the number of nodes $v'$ which are at greater depth than $v$ and which have $i \in Y_{L(v')}$. Define $W_v = \{ (i, u_{i,v} + 1) \mid i \in Y_{L(v)} \}$.
\end{proof}

 Lemma~\ref{r1fwitness-tree-lemma} follows from Lemma~\ref{r1fwitness-tree-lemma2}; as the entries of $R$ are fair coins, the probabilities of each event $f_{L(v)}(\pi_{W_v}(R)) = 1$ is $P_{\Omega}(L(v))$; furthermore, since the sets $W_v$ are non-intersecting, these events are all independent.

\subsection{Derandomizing Moser-Tardos}
The original paper of Moser \& Tardos gave a sequential deterministic algorithm that only worked for a very limited class of LLL instances, for example when $d$ was constant. An NC algorithm was later given in \cite{mt2}, covering a slightly larger class of bad-events. This algorithm required satisfying the LLL criterion with a slack, in particular it required $e p d^{1+\epsilon} \leq 1$ for some $\epsilon > 0$, and had a complexity of roughly $(\frac{\log^3 mn}{\epsilon}, m^{O(1/\epsilon)})$. An alternative NC algorithm was provided in \cite{mt3}, which is slightly faster than \cite{mt2}. 

These latter algorithms have numerous conditions on the functions $f_B$; roughly speaking, they require $f_B$ to have decision-tree complexity of order $\log d$. The clearest example of this problem type is $k$-SAT, in which bad-event corresponds to a clause being violated. So each bad-event is defined by $x_{i_1} = j_1 \wedge \dots  \wedge x_{i_k} = j_k$, a monomial in $k$ variables. 

Many other LLL applications, particularly those in which the bad-events are determined by sums of random variables, do not fit into this paradigm; we discuss two examples in Section~\ref{r1fmt-example}. The hallmark of these types of problems is that the bad-events are \emph{complex} boolean functions; our focus here will be to give NC algorithms for such problems.

The analysis of \cite{mt3} is based on an extension of the witness tree to a more general object referred to as a \emph{collectible witness DAG (CWD)}. These objects represent in a sense all the ways the MT algorithm could require a long execution time. This requires a great deal of notation to define properly, but the important point for us is that each CWD $\tau$ satisfies Lemma~\ref{r1fwitness-tree-lemma2} in the same a witness tree does.  We will not discuss the (technical) differences between witness trees and CWD's.

We say that a CWD $\tau$ is \emph{compatible} with a resampling table $R$ if condition (A2) of Lemma~\ref{r1fwitness-tree-lemma2} is satisfied, namely $f_{L(v)}(\pi_{W_v}(R)) = 1$ for every $v \in \tau$. For any set $\mathcal T$ of CWD's and a resampling table $R$,  we define $\mathcal T^R \subseteq \mathcal T$ to be the set of CWD's $\tau \in \mathcal T$ compatible with $R$. We summarize some key results of \cite{mt3} which are relevant to us.
\begin{lemma}[\cite{mt3}]
\label{r1flem3}
Suppose that $e p d^{1+\epsilon} \leq 1$ for $\epsilon > 0$, and suppose that the functions $f_B$ can be evaluated with complexity $(U, \poly(m,n))$ where $m = | \mathcal B |$.

Let $K = \frac{c \log (mn/\epsilon)}{\epsilon \log d}$ for some constant $c > 0$. There is a set $\mathcal T$ of CWD's with the following properties:
\begin{enumerate}
\item[(T1)] $|\mathcal T| \leq (m n/\epsilon)^{O(1/\epsilon)}$
\item[(T2)] Each $\tau \in \mathcal T$ has size at most $2 K$.
\item[(T3)] If every $\tau \in \mathcal T^R$ has size less than  $K$, then an assignment avoiding $\mathcal B$ can be found with complexity $(K U + K \log(mn |\mathcal T^R|) + \log^2 |\mathcal T^R|, \poly(K, m,|\mathcal T^R|))$.
\item[(T4)] $\sum_{\tau \in \mathcal T, |\tau| \geq K} w(\tau) < 1/2$.
\item[(T5)] $\sum_{\tau \in \mathcal T} w(\tau) < O(m)$.
\item[(T6)] The set $\mathcal T$ can be enumerated with quasi-complexity $( \frac{\log^2(mn/\epsilon)}{\epsilon}, (m n/\epsilon)^{O(1/\epsilon)})$.
\end{enumerate}
\end{lemma}

Now consider drawing a resampling table $R(i,j)$ where $i = 1, \dots, n$ and $j = 1, \dots, 2 K$. Given any CWD $\tau$, let $f(\tau, R)$ denote the indicator function that $\tau$ is compatible with $R$. We also define
\begin{equation}
\label{sen1}
S(R) = \frac{1}{C m} \sum_{\tau \in \mathcal T} f(\tau,R) + \sum_{\substack{\tau \in \mathcal T \\ |\tau| \geq K}} f(\tau,R)
\end{equation}
for some constant $C > 0$.

When the entries of $R$ are drawnly independent from $\mathcal M_b$, properties (T4) and (T5) ensure that $\bE[S(R)] < 1$ for a sufficiently large $C$. If $S(R) < 1$, the resulting resampling table $R$ has the property that $| \mathcal T^R | \leq O(m)$ and that every $\tau \in \mathcal T^R$ has size less than $K$. Therefore, by property (T3), a satisfying assignment can be found in  quasi-complexity $(\frac{U \log(mn) + \log^2(mn)}{\epsilon}, \text{poly}(m, n, 1/\epsilon))$. Thus, the problem of finding a satisfying assignment is reduced to the problem of minimizing $S(R)$, which is a sum of juntas. 

\begin{theorem}
\label{r1four-det-thm}
Suppose that $e p d^{1+\epsilon} \leq 1$ for some $\epsilon > 0$. Let $m = |\mathcal B|$.  Suppose that each bad-event $B \in \mathcal B$ is a $w$-junta (i.e. $|Y_B| \leq w$) on $\mathcal M_b^n$.

Suppose that we have a graded PEO for the collection of functions $f_B$ which has complexity $(C_1, \poly(m, n))$; namely, given any $B \in \mathcal B$ as well as a graded partial assignment to the variables in $Y_B$, it computes the corresponding probability that $f_B = 1$.

Then we can find $x \in \mathcal M_b^n$ avoiding $\mathcal B$ in quasi-complexity $(w b (C_1 + \frac{\log mn}{\epsilon}) + \frac{C_1 \log mn}{\epsilon}, (mn)^{O(1/\epsilon)})$.
\end{theorem}
\begin{proof} 
If $\epsilon < 1/(mn)$, then we can solve this problem by exhaustive search in $2^n \leq (m n)^{O(1/\epsilon)}$ processors. So let us assume that $\epsilon \geq 1/(mn)$. We also assume $C_1 \geq \Omega(\log mn)$ as this time is required to read the input.

Let $m' = (m n)^{c/\epsilon}$. For an appropriate constant $c$, the objective function $S(R)$ is by (T1) a sum of at most $m'$ functions $f(\tau, R)$. By (T2), each $\tau \in \mathcal T$ has size at most $2 K$; property (A4) ensures that each term $f_{L(v)}(\pi_{W_v}(R))$ depends on at most $w$ entries of $R$, so in all each function $f(\tau, R)$ is a $w'$-junta for some $w' = O(\epsilon^{-1} w \log mn)$. The total number of variables determining $R$ is $n' = 2 n K \leq O(\epsilon^{-1} n \log(mn))$. (We do not need to compute any entries of $R$ beyond this point.)

We claim next that we can form a graded PEO for the functions $f(\tau, R)$ with  quasi-complexity $(C'_1, C'_2)$, where $C_1' = C_1 + \frac{\log mn}{\epsilon}$ and $C_2' = (mn)^{O(1/\epsilon)}$. For suppose we are given a graded partial assignment query $R'$. We can compute the associated projections $\pi_{W_1}(R'), \dots, \pi_{W_v}(R')$ for each $\tau$ using $\tilde O(\frac{\log mn}{\epsilon})$ time. The probability that any $\tau$ is compatible with $R'$ is simply the product of the probabilities of $f_{B_i}(\pi_{W_i}(R))$. The PEO for $f_B$ allows us to compute these probabilities in parallel with quasi-complexity $( (\log K) C_1, (mn)^{O(1/\epsilon)})$. Finally, we multiply the probabilities together in $\tilde O(\log K)$ time.

Now apply Theorem~\ref{r1fpolylog-thm} to find $R$ with $S(R) \leq \bE[S(R)] < 1$ using quasi-complexity $\tilde O( b C_1' (1 + \frac{w'}{\log m'n'}) , (w')^{O(1)} C_2')$, which simplifies to $(w b (C_1 + \frac{\log mn}{\epsilon}), (mn)^{O(1/\epsilon)} )$.  Once we have found $R$, we use (T3) to find $x$ avoiding $\mathcal B$. The PEO can be used to check whether a given bad-event is true, so $U \leq C_1$ and this step requires $O(\frac{C_1 \log mn}{\epsilon})$ time. 
\end{proof}

\section{Applications of the LLL}
\label{r1fmt-example}

\subsection{Defective vertex coloring}
A \emph{$k$-defective vertex $c$-coloring} of a graph $G = (V,E)$, is an assignment of colors to the vertices such that every vertex $v$ has at most $k$ neighbors with the same color as $v$. This generalizes proper vertex coloring, in that a proper vertex coloring is a $0$-defective coloring. In this section, we give an algorithm which gives a $k$-defective $c$-coloring of a graph $G$ of maximum degree $\Delta$ with $c = O(\Delta/k)$, for any choice of $k$ in the range $\{1, \dots, \Delta \}$. The main idea, inspired by a similar randomized distributed algorithm of \cite{ghaffari}, is a degree-splitting step; when $\Delta$ is small, this can be achieved efficiently using our deterministic LLL algorithm and when $\Delta$ is large then we can use an alternate algorithm based on simple Chernoff bounds.

\begin{lemma}
\label{splitprop1}
There is an absolute constant $K$ with the following property. Given a graph $G$ of maximum degree $\Delta$ and an integer parameter $j  \leq \log_2 (\frac{\Delta}{K \log \Delta})$, there is an algorithm with quasi-complexity $(\Delta \log n + \log^2 n, \poly(n))$ to $2^j$-color the vertices, so that each vertex $v$ has at most $(\Delta/2^j) (1 + K \sqrt{(2^j/\Delta) \log \Delta})$ neighbors with the same color as $v$.
\end{lemma}
\begin{proof}
Consider applying the LLL to the random process in which each vertex independently and uniformly selects a color (represented as a $j$-bit string). Each vertex $v$ has a bad-event $B_v$ that it has too many neighbors of its own color. Thus there are $m = n$ bad-events, and each bad-event involves at most $\Delta+1$ variables. The number of neighbors of each color  is a binomial random variable with mean at most $\Delta 2^{-j}$.  Note that $B_v \sim B_w$ iff $v$ and $w$ are at distance at most 2 in $G$. So in the sense of the LLL we have $d \leq \Delta^2$. 

Let $\delta = c \sqrt{ (2^j/\Delta) \log \Delta_i }$ for some constant $c$. For $K$ sufficiently large, we ensure that $\delta \leq 1$. Therefore, by the Chernoff bound, $B_v$ has probability at most $e^{-\mu \delta^2/3}$, which is smaller than $\Delta^{-4}$ for an appropriate choice of $c$. So, in the sense of the LLL, we have $p \leq \Delta^{-4}$. These parameters satisfy Theorem~\ref{r1four-det-thm} with $\epsilon = 1/2$. Each bad-event $B_v$ is a boolean function on at most $\Delta$ variables, and a graded PEO can be constructed with quasi-complexity $(\log n, \poly(n))$. Thus Theorem~\ref{r1four-det-thm} gives the desired goal in quasi-complexity $( j \Delta \log n + \log^2 n, \poly(n))$.  Note that $j \leq O( \log \Delta)$, and so it can be dropped from the quasi-complexity bounds.
\end{proof}

\begin{theorem}
\label{defective-thm}
Let $G$ be a graph with maximum degree $\Delta$ and $k \in \{1, \dots, \Delta \}$. Then there is an NC algorithm running in time $\tilde O(\log^2 n)$ to obtain a $k$-defective vertex $c$-coloring with $c = O(\Delta/k)$.
\end{theorem}
\begin{proof}
When $\Delta \geq \log n$, let us consider the random process of assigning every vertex a color uniformly selected from $\frac{\Delta}{\log n}$; a simple Chernoff bound shows that, with high probability, this ensures that each vertex has at most $C \log n$ neighbors of each color class, where $C > 0$ is some sufficiently large constant. This can be derandomized by an algorithm of Sivakumar \cite{sivakumar} (among other methods), as there are a polynomial number of ``statistical tests'' (in this case, the degree of each vertex with respect to each color class) which can be computed in logspace. After this first coloring step, which can be executed in $O(\log^2 n)$ time, we get multiple subgraphs with maximum degree $\Delta \leq C \log n$.

Thus we can assume that $\Delta \leq C \log n$ for a constant $C > 0$. In this case we can use iterated applications of Lemma~\ref{splitprop1}. Each iteration reduces the degree of the residual graphs by a logarithmic factor, and so the overall running time is close to the running time of a single application of Lemma~\ref{splitprop1}. We defer the full proof to Appendix~\ref{def-thm-app}, as the construction is technical and similar to that of \cite{ghaffari}.
\end{proof}

\subsection{Domatic partition}
A \emph{domatic partition} of a graph is a $c$-coloring of the vertices of $G$ with the property that each vertex of $G$ sees all $c$-colors in its neighborhood (including itself). That is, for any color $\ell = 1, \dots, c$, the color-$\ell$ vertices form a dominating set of $G$. An algorithm was given in \cite{domatic} using the LLL to find a domatic partition with a large number of colors. For simplicity, we specialize their algorithm to $k$-regular graphs.
\begin{theorem}
Let $\eta > 0$ be any fixed constant. There is some constant $K = K_{\eta}$ with the following property.  If $G$ is $k$-regular with $k > K$, then $G$ has a domatic partition of size  $c \geq (1-\eta) \frac{k}{\log k}$, which can be found using $\tilde O_{\eta}(k \log n + \log^2 n)$ time and $n^{O_{\eta}(1)}$ processors. 
\end{theorem}
\begin{proof}
We follow the iterated LLL construction of \cite{domatic}, in which the color of each vertex is an ordered pair $\chi(v) = ( \chi_1(v), \chi_2(v) )$; here $\chi_1$ is chosen from $c_1 = k/\log^3 k$ colors, and $\chi_2$ is chosen from $c_2 = (1 - \eta) \log^2 k$ colors. In the first phase of the LLL, we will select $\chi_1$ and the second phase will select $\chi_2$. Each vertex chooses its colors uniformly at random among $[c_1], [c_2]$ respectively.\footnote{If $c_1, c_2$ are not powers of two, then we cannot directly represent this in our bit-based LLL formulation. However, we can simulate it by drawing values $u_1, u_2$ from $[2^{r_1}, 2^{r_2}]$ for $r_i = \log_2 c_i + O(\log k)$, and projecting uniformly down to $[c_i]$. This changes the probabilities of the bad-events by a negligible factor of $\text{polylog}(1/k)$. We omit further details for simplicity.}

Now consider the phase I coloring. For each vertex $v$ and each color $j \in [c_1]$, define $N_{j}(v)$ to the set of neighbors $w$ with $\chi_1(w) = j$ and let $X_{v,j} = |N_j(v)|$. The expected value of $X_{v,j}$ is $\mu = \log^3 k$. For each vertex $v$ and each color $j \in [c_1]$, we have a bad-event $B_{v,j}$ that $X_{v,j} \leq t_0$ or $X_{v,j} \geq t_1$, where $t_0 = \mu - \phi \log^2 k$ and $t_1 = \mu + \phi \log^2 k$ and $\phi$ is a large constant. 

For $\phi$ sufficiently large, the Chernoff bound shows that $B_{v,j}$ has probability at most $p \leq k^{-5}$. Furthermore, each bad-event $B_{v,j}$ affects $B_{v',j'}$ only if $\text{dist}(v,v') \leq 2$, so in the sense of the LLL we have $d \leq k^4$. A graded PEO for these for these bad-events with running time $C_1 = \tilde O(\log mn)$. Apply Theorem~\ref{r1four-det-thm} to find $\chi_1$ with quasi-complexity $(k \log n + \log^2 n, n^{O(1)})$. 

For each vertex $v$, each $j \in [c_1]$, and each $j' \in [c_2]$, we have a bad-event $B_{v,j, j'}$ that there is no $w \in N_j(v)$ with $\chi_2(w) = j'$; if all such bad-events are avoided then the resulting coloring $(\chi_1(v), \chi_2(v))$ gives a domatic partition. The only dependencies now are between bad-events $B_{v, j, j'}$ and $B_{w,j,j''}$ where $v, w$ share a neighbor $u$ with $\chi_1(u) = j$, so  $d \leq t_1 k c_2$ and $p \leq (1 - 1/c_2)^{t_0}$.

Set $\epsilon = \eta/2, \phi = 10$. It is straightforward to verify that the criterion $e p d^{1+\epsilon} \leq 1$ is satisfied when $k$ is sufficiently large. Thus, Theorem~\ref{r1four-det-thm} gives a coloring avoiding the phase-II bad-events using $\tilde O_{\eta}(k \log n + \log^2 n)$ time and $n^{O_{\eta}(1)}$ processors.
\end{proof}

\section{Acknowledgments}
Thanks to Aravind Srinivasan, Vance Faber, and anonymous referees for helpful comments and discussion.

\appendix

\section{The PEO for hypergraph rainbow coloring} 
\begin{proposition}
\label{r1fgraded-rainbow-oracle}
Let $f_e$ be the indicator function that edge $e$ is rainbow.  Then the collection of functions $f_e$ has a graded PEO with overall quasi-complexity $(\log mn, m+n)$.
\end{proposition}
\begin{proof}
It suffices to compute the probability that a given edge $e$ will be rainbow on the coloring $\phi_x$ for some graded $u \in \{0, 1, \text{?} \}^{d b}$ (here $u$ represents the projection of the overall partially-graded $x \in \{0, 1,  \text{?} \}^{n b}$ to the vertices in $e$). Since $d b = (mn)^{o(1)}$, the processor complexity of this task can be an arbitrarily polynomial in $b, d$.

We first describe how to do so if $u$ is fully-graded; we then modify it to allow $u$ to be merely graded. Suppose the most-significant $\ell \leq b$ bit-levels of the vector $y$ have been determined and the least-significant $b - \ell$ bit-levels of $u$ remain fair coins.  We may write $u$ in the form $u_v = (y_v, \text{?}, \dots, \text{?})$, where $y_v \in \mathcal M_{\ell}$. For each $c \in \mathcal M_{\ell}$ let $S_c$ denote the set of vertices $v \in e$ with $y_v = c$.

For each  $k \in \{0, \dots, d-1 \}$, let us define $G_k$ to be the set of values $c \in \{0, \dots, 2^{\ell} - 1 \}$ such that some vertex $v \in S_c$ could (depending on the lower order bits of $y$) be assigned color $k$. Specifically, 
$$
G_k = \{ c \in \{0, \dots, 2^{\ell} - 1 \} \mid F(2^{b-\ell} c) \leq k \leq F( 2^{b-\ell} (c+1) - 1 ) \}
$$
where the function $F: \{0, 2^b - 1 \} \rightarrow \{0, \dots, d - 1\}$ is defined by $F(x) = \lfloor (d/2^b) x \rfloor$.

Observe that if $y \geq x + (d/2^b)$, then we must have $F(y) > F(x)$. Using this fact, we claim that each $G_k$ is either a singleton set, or a set of two adjacent elements $\{c, c+1 \}$. For, suppose not; then there must exist $c_1, c_2 \in G_k$ with $c_2 > c_1 + 1$ and
\begin{eqnarray*}
&F(2^{b-\ell} c_1) \leq k \leq F(2^{b-\ell} (c_1+1) - 1) \\
&F(2^{b-\ell} c_2) \leq k \leq F(2^{b-\ell} (c_2+1) - 1)
\end{eqnarray*}

But, note that in this case
\begin{align*}
(2^{b-\ell} c_2) - ( 2^{b-\ell} (c_1+1) - 1) &= 2^{b-\ell} (c_2 - c_1 - 1) + 1 \geq 2^{b-\ell} + 1 \geq 1 \geq (d/2^b)
\end{align*}
and so $F( 2^{b-\ell} c_2 ) > F(2^{b-\ell} (c_1+1) - 1)$, a contradiction.

Also, we claim that for each value of $c$, there is at most one value $k$ such that $G_k = \{c, c+1 \}$. For, if not, then there would be values $k_1 < k_2$ with
\begin{eqnarray*}
F(2^{b-\ell} (c+1)) &\leq k_1 &\leq F(2^{b-\ell}(c+2) - 1) \\
F(2^{b-\ell} c) &\leq k_2 & \leq F(2^{b-\ell}(c+1) - 1)
\end{eqnarray*}

But note then that
$$
F(2^{b-\ell} (c+1) - 1) \geq k_2 > k_1 \geq F(2^{b-\ell} (c+1)) > F(2^{b-\ell} (c+1)) \geq F(2^{b-\ell} (c+1) - 1),
$$
a contradiction.

Thus, for each $c \in \{0, \dots, 2^{\ell} - 1\}$, let us define $W_c$ to be the set of values $k \in \{0, \dots, d-1 \}$ such that $G_k = \{c, c + 1 \}$. We have just showed that $|W_c| \leq 1$.

Now consider the random experiment of assigning independent Bernoulli-$1/2$ values to the low-order $b - \ell$ bit-levels of $u$. Define the random variable $Z_c$ to be the number of vertices in $S_c$ assigned a value $k \in W_c$. (If $W_c = \emptyset$, then $Z_c = 0$ necessarily.) In order for $e$ to be rainbow, every $c \in \mathcal M_{\ell}$ must have $Z_c \in \{0, 1 \}$.

For any integers $0 \leq c_0 < c_1 \leq 2^\ell$ and values $z_0, z_1 \in \{0, 1 \}$, let us thus define the function $g$ by
\begin{align*}
&g(c_0, c_1, z_0, z_1) = \\
& \qquad P( \text{the vertices in $S_{c_0}, S_{c_0+1}, \dots, S_{c_1 - 1}, S_{c_1}$ receive distinct colors and $Z_{c_1} = z_1$} \mid Z_{c_0-1} = z_0)
\end{align*}

The overall probability that the random experiment results in a rainbow coloring of $e$ is given by $g(0,2^\ell,0,0)$. With a little thought, one can see that $g$ satisfies the recurrence:
$$
g(c_0, c_1, z_0, z_1) = g(c_0, c_2 - 1, z_0, 0) g(c_2, c_1, 0, z_1) +  g(c_0, c_2 - 1, z_0, 1) g(c_2, c_1, 1, z_1)
$$
for $c_2 = (c_0 + c_1)/2$. 

If $S_c = \emptyset$, then the value $c$ is not relevant to this calculation; thus, during this calculation, we can skip all such entries. As there are at most $d$ values of $c$ with $S_c \neq \emptyset$, we can recursively compute $g(0,2^\ell,0,0)$ using $\poly(db)$ processors and using $\tilde O(\log mn)$ time. (The base cases can be computed using simple arithmetic as functions of $|S_c|$.)

We next discuss how to modify this to graded PEO. Here, the top $\ell-1$ bits of each $y_v$ are completely known, while the lowest-order bit is in $\{0, 1, \text{?} \}$. Now suppose we want to calculate $g(0, 1, z_0, z_1)$; in this case, some vertices are known to correspond to the sets $S_0, S_1$ and some vertices (for which bit at level $\ell$ is unspecified) have a $1/2$ probability of going into $S_0$ and a $1/2$ probability of going into $S_1$. We can integrate over the sizes of $S_0$ and $S_1$ (which are now binomial random variables ), and use the above formulas to calculate $g(0,1,z_0, z_1)$.
\end{proof}

\section{Full proof of Theorem~\ref{defective-thm}}
\label{def-thm-app}

We assume here that $\Delta \leq C \log n$ for some sufficiently large constant $C$, and that $k$ is larger than any needed constant. The case of large $\Delta$ has already been discussed.

We will build the coloring gradually over stages $i = 0, \dots, r+1$; at stage $i$, the vertices have a $t_i$-coloring, in which  every vertex has at most $\Delta_i$ neighbors of its own color class.  At stage $i$, and for any integer $\ell \in \{1, \dots, t_i \}$, let us define $G_{i,\ell}$ to be the subgraph of $G$ induced on vertices with color $\ell$. So $G_{i,\ell}$ has maximum degree $\Delta_i$. We will apply Lemma~\ref{splitprop1} with parameter $j_i$ to each $G_{i,\ell}$; this is done in parallel across all values of $\ell$. This gives $t_{i+1} \leq 2^{j_i} t_i$, and at the end of this process, we thereby obtain a $\Delta_{r+1}$-defective coloring with $t_{r+1}$ colors.

We need to define the sequence of values $j_i, \Delta_i$ which will be valid for the degree splitting procedure. We do so recursively by setting $\Delta_0 = \Delta$, and
$$
\Delta_{i+1} = (\Delta_i/2^{j_i}) (1 + K \sqrt{2} \log^{-1/2} \Delta_i) \qquad j_i = \begin{cases}
\lceil \log_2 \bigl(\Delta_i/\log^2 \Delta_i\bigr) \rceil & i < r \\
\lceil \log_2 \bigl(\Delta_i/k\bigr) \rceil & i = r
\end{cases}
$$
where $r$ is a parameter to be determined. 

Let us verify that these parameters satisfy the conditions of Lemma~\ref{splitprop1}; specifically, we show by induction on $i$ that  $G_i$ has maximum degree $\Delta_i$ and $2^{j_i} \leq \frac{\Delta_i}{k \log \Delta_i}$ for $i = 0, \dots, r$.

In order to carry out this analysis, let us define a sequence of real numbers by
$$
b_0 = \Delta,  \qquad b_{i+1} = \tfrac{1}{2} \log^2 b_i
$$
Let $r$ be the largest integer with $b_{r} \geq k$; we stop this procedure at stage $r+1$. We easily see that $r \leq O(\log^* \Delta)$.   We claim that for $i = 0, \dots, r$ we have $k \leq b_i \leq \Delta_i \leq 4 b_i$, and we show this by induction on $i$. The bound $b_i \geq k$ is immediate from the definition of $r$. The bound on $\Delta_0$ is immediate. For $i < r$, the lower bound is shown by
\begin{align*}
\Delta_{i+1} &\geq \Delta_i/2^{j_i} \geq \Delta_i/2^{\lceil \log_2 (\Delta_i/\log^2 \Delta_i) \rceil} \geq \Delta_i/2^{1 + \log_2 (\Delta_i/\log^2 \Delta_i)} = \tfrac{1}{2} \log^2 \Delta_i \geq \tfrac{1}{2} \log^2 b_i = b_{i+1}
\end{align*}

For the upper bound, we have for $b_i \geq k$ and $k$ sufficiently large,
\begin{align*}
\Delta_{i+1} &= (\Delta_i / 2^{j_i}) (1 + K \sqrt{2} \log^{-1/2} \Delta_i) \leq (1 + K \sqrt{2} \log^{-1/2} \Delta_i) \log^2 \Delta_i \\
&\leq (1 + K \sqrt{2} \log^{-1/2} b_i) \log^2 (4 b_i) \leq 1.01 \log^2 (b_i)  \leq 4 b_{i+1}
\end{align*}

We can now show that $2^{j_i} \leq \frac{\Delta_i}{K \log \Delta_i}$ holds. For $i < r$, we have
\begin{align*}
2^{j_i} &\leq \frac{2 \Delta_i}{\log^2 \Delta_i} \leq \frac{\Delta_i}{\log \Delta_i} \times \frac{2}{\log \Delta_i} \leq \frac{\Delta_i}{\log \Delta_i} \times \frac{2}{k} \leq \frac{\Delta_i}{K \log \Delta_i}
\end{align*}

For $i = r$, we note that $b_{r+1} \leq k$ and so 
\begin{align*}
2^{j_r} &\leq 2 \Delta_r/ k \leq (8 b_r) / k \leq 16 b_r / \log^2 b_r \leq 16 \Delta_r / \log^2 \Delta_r \leq \Delta_r / (K \log \Delta_r) 
\end{align*}

So we can apply Lemma~\ref{splitprop1}. For $i < r$, the definition of $j_i$ gives 
$$
2^{j_i}/\Delta_i \leq \frac{2}{\log^2 \Delta_i}
$$
and therefore Lemma~\ref{splitprop1} shows that the graph $G_{i+1, \ell}$ has maximum degree
at most
\begin{align*}
(\Delta_i/2^{j_i}) (1 + K \sqrt{ (2^{j_i}/\Delta_i) \log \Delta_i } ) &\leq ( \Delta_i/2^{j_i}) (1 + K \sqrt{ \frac{2}{\log \Delta_i}}) = \Delta_{i+1}
\end{align*}

Similarly, for $i = r$,  Lemma~\ref{splitprop1} ensures that $G_{r+1, \ell}$ have maximum degree
$$
\Delta_{r+1} \leq (\Delta_r / 2^{j_r}) (1 + O(\log^{-1/2} \Delta_r)) \leq k (1 + O(\log^{-1/2} \Delta_r)) \leq O(k)
$$

Thus, the overall coloring we obtain is indeed $O(k)$-defective. Our next task is to count the number of colors used. Let us define $a_i = \Delta_i t_i$. We want to show that $t_{r+1} \leq O(\Delta/k)$. As $\Delta_{r+1} \geq k/2$, it suffices to show that $a_{r+1} \leq O(\Delta)$. The recursive formulas for $\Delta_i$ and $t_i$ give
\begin{align*}
a_{i+1} &\leq (\Delta_i / 2^{j_i}) (1 + K \sqrt{2} \log^{-1/2} \Delta_i) \times (t_i 2^{j_i}) \leq a_i (1 + K \sqrt{2} \log^{-1/2} \Delta_i) \leq a_i (1 + K \sqrt{2} \log^{-1/2} b_i)
\end{align*}

Therefore,
$$
a_{r+1} \leq a_0 \prod_{i=0}^r (1 + K \sqrt{2} \log^{-1/2} b_i) \leq \Delta e^{O( \sum_{i=0}^r \log^{-1/2} b_i)} \leq O(\Delta)
$$
where the last inequality follows by noting that the sequence $\log b_i$ is decreasing super-exponentially.

We finish by calculating the complexity of the algorithm. Each iteration $i$ requires $(\Delta_i \log n + \log^2 n, \poly(n))$ quasi-complexity. We see easily that $\Delta_i \leq O(\Delta) \leq O(\log n)$, so this is $\tilde O(\log^2 n)$ time per iteration. There are $r \leq O(\log^* \Delta)$ iterations, giving a total runtime again of $\tilde O(\log^2 n)$.


\begin{thebibliography}{1}

\bibitem{alon-srin}
Alon, N., Srinivasan, A.: Improved parallel approximation of a class of integer programming problems. Proceedings of the 23rd International Colloquium on Automata, Languages, and Programming (ICALP), pp. 562-573 (1996)

\bibitem{alon-babai}
Alon, N., Babai, L., Itai, A.: A fast and simple randomized parallel algorithm for the maximal independent set problem. Journal of Algorithms 7, pp. 567-583 (1986)

\bibitem{berger-simulating}
Berger, B., Rompel, J.: Simulating $(\log^c n)$-wise independence in NC. Journal of the ACM 38-4, pp. 1026-1046. (1991)

\bibitem{borodin-cook}
Borodin, A., Cook, S.: A time-space tradeoff for sorting on a general sequential model of computation.  SIAM Journal on Computing 11-2, pp. 287-297 (1982)

\bibitem{mt2}
Chandrasekaran, K., Goyal, N., Haeupler, B.: Deterministic algorithms for the Lov\'{a}sz local lemma. SIAM Journal on Computing 42-6, pp. 2132-2155 (2013)

\bibitem{crs}
Chari, S., Rohatgi, P, Srinivasan, A.: Improved algorithms via approximations of probability distributions. Journal of Computer and System Sciences 61-1, pp. 81-107 (2000)

\bibitem{ghaffari}
Fischer, M., Ghaffari, M.: Sublogarithmic distributed algorithms for Lov\'{a}sz local lemma, and the complexity hierarchy. Proceedings of the 31st International Symposium on Distributed Computing (DISC), pp. 18:1-18:16 (2017)

\bibitem{domatic}
Feige, U., Halld\'{o}rsson, M. M., Kortsarz, G., Srinivasan, A.: Approximating the domatic number. SIAM Journal on Computing 32-1, pp. 172-195 (2002)

\bibitem{mt3}
  Haeupler, B., Harris, D.G.: Parallel algorithms and concentration bounds for the Lov\'{a}sz Local Lemma via witness-DAGs. ACM Transactions on Algorithms 13-4, Article \#53 (2017)
  
\bibitem{harris2}
  Harris, D.: Deterministic parallel algorithms for bilinear objective functions. arxiv:1711-08494. (2017)
  
\bibitem{luby-old}
Luby, M.: Removing randomness in parallel computation without a processor penalty. Journal of Computer and System Sciences 47-2, pp. 250-286 (1993)

\bibitem{mt}
Moser, R., Tardos, G.: A constructive proof of the general Lov\'{a}sz Local Lemma. Journal of the ACM 57-2, pp.\ 11:1-11:15 (2010)

\bibitem{nc-gauss}
Mulmuley, K. A fast parallel algorithm to compute the rank of a matrix over an arbitrary field. Combinatorica 7-1, pp. 101-104 (1987).


\bibitem{naor2}
Naor, J., Naor, M.: Small-bias probability spaces: efficient construction and applications, SIAM Journal of Computing 22-4, pp. 835-856 (1993)

\bibitem{schulman}
Schulman, L.: Sample spaces uniform on neighborhoods. Proceedings of the 24th ACM Symposium on Theory of Computing (STOC), pp. 17-25 (1992)

\bibitem{sivakumar}
Sivakumar, D.: Algorithmic derandomization via complexity theory. Proceedings of the 34th ACM Symposium on Theory of Computing (STOC), pp. 619-626 (2002)


\end{thebibliography}
\end{document}